\documentclass[conference]{IEEEtran}
\usepackage{graphicx}

\usepackage{epsfig,endnotes}
\usepackage{todonotes}
\usepackage{latexsym} 
\usepackage{subfig}
\usepackage{algorithm}
\usepackage[noend]{algpseudocode}
\usepackage{listings}
\usepackage{empheq}

\usepackage{amsmath}
\usepackage{amssymb}

\usepackage{amsthm}
\usepackage{amssymb}

\usepackage{multirow}
\usepackage{bigdelim}
\usepackage{colortbl}
\makeatletter
\def\BState{\State\hskip-\ALG@thistlm}
\makeatother
\usepackage{hyperref}

\usepackage{mathtools}
\usepackage{paralist}
\usepackage{hyphenat}

\usepackage{tikz}
\usetikzlibrary{shapes}
\usetikzlibrary{calc,shapes.multipart,chains,arrows,positioning} 

\newcommand{\chead}{\cellcolor{tableHeadGray}}

\newcommand{\parttitle}[1]{\vspace{-.3em}\noindent\textbf{#1.}}

\newcommand{\cList}{\thickspace \| \thickspace}

\newcommand{\mathtext}[1]{\thickspace\text{#1}\thickspace}

\newcommand{\projection}{\Pi}
\newcommand{\selection}{\sigma}
\newcommand{\aggregation}{\gamma}
\newcommand{\union}{\cup}
\newcommand{\intersection}{\cap}
\newcommand{\difference}{-}

\newcommand{\duplicate}{\delta}

\newcommand{\join}{\bowtie}
\newcommand{\crossprod}{\times}

\newcommand{\schema}[1]{\textsc{Sch}(#1)}

\newcommand{\eIf}[3]{\text{if}\thickspace #1 \thickspace
  \text{then} \thickspace #2 \thickspace \text{else} \thickspace #3}

\newcommand{\stopThresh}{\tau}

\newcommand{\ancestor}{\mathop{\stackrel{*}{\leadsto}}}
\newcommand{\allAncestor}{\mathop{\stackrel{\forall}{\leadsto}}}

\definecolor{black}{rgb}{0,0,0}
\definecolor{grey}{rgb}{0.8,0.8,0.8}
\definecolor{red}{rgb}{1,0,0}
\definecolor{green}{rgb}{0,1,0}
\definecolor{darkgreen}{rgb}{0,0.5,0}
\definecolor{darkpurple}{rgb}{0.5,0,0.5}
\definecolor{darkdarkpurple}{rgb}{0.3,0,0.3}
\definecolor{blue}{rgb}{0,0,1}
\definecolor{shadegreen}{rgb}{0.95,1,0.95}
\definecolor{shadeblue}{rgb}{0.95,0.95,1}
\definecolor{shadered}{rgb}{1,0.85,0.85}
\definecolor{oddRowGrey}{rgb}{0.95,0.95,0.95}
\definecolor{evenRowGrey}{rgb}{0.85,0.85,0.85}
\definecolor{tableHeadGray}{rgb}{0.85,0.85,0.85}

\newtheorem{Theorem}{Theorem}
\newtheorem{Definition}{Definition}

\newtheorem{Example}{Example}
\makeatletter
\renewcommand{\ALG@beginalgorithmic}{\footnotesize}
\makeatother

\newcommand{\TRtitle}{Optimizing Provenance Computations}
\newcommand{\TRauthors}{Xing Niu and Boris Glavic}
\newcommand{\TRnumber}{IIT/CS-DB-2016-02}
\newcommand{\TRdate}{2016-10}

\begin{document}

\lstdefinestyle{psql}
{
tabsize=2,
basicstyle=\small\upshape\ttfamily,
language=SQL,
morekeywords={PROVENANCE,BASERELATION,INFLUENCE,COPY,ON,TRANSPROV,TRANSSQL,TRANSXML,CONTRIBUTION,COMPLETE,TRANSITIVE,NONTRANSITIVE,EXPLAIN,SQLTEXT,GRAPH,IS,ANNOT,THIS,XSLT,MAPPROV,cxpath,OF,TRANSACTION,SERIALIZABLE,COMMITTED,INSERT,INTO,WITH,SCN,UPDATED},
extendedchars=false,
keywordstyle=\color{blue},
mathescape=true,
escapechar=@,
sensitive=true
}

\lstdefinestyle{datalog}
{
basicstyle=\footnotesize\upshape\ttfamily,
language=prolog
}

\lstdefinestyle{rsl}
{
tabsize=3,
basicstyle=\small\upshape\ttfamily,
language=C,
morekeywords={RULE,LET,CONDITION,RETURN,AND,FOR,INTO,REWRITE,MATCH,WHERE},
extendedchars=false,
keywordstyle=\color{blue},
mathescape=true,
escapechar=@,
sensitive=true
}

\lstdefinestyle{pseudocode}
{
  tabsize=3,
  basicstyle=\small,
  language=c,
  morekeywords={if,else,foreach,case,return,in,or},
  extendedchars=true,
  mathescape=true,
  literate={:=}{{$\gets$}}1 {<=}{{$\leq$}}1 {!=}{{$\neq$}}1 {append}{{$\listconcat$}}1 {calP}{{$\cal P$}}{2},
  keywordstyle=\color{blue},
  escapechar=&,
  numbers=left,
  numberstyle={\color{green}\small\bf}, 
  stepnumber=1, 
  numbersep=5pt,
}

\lstdefinestyle{xmlstyle}
{
  tabsize=3,
  basicstyle=\small,
  language=xml,
  extendedchars=true,
  mathescape=true,
  escapechar=£,
  tagstyle={\color{blue}},
  usekeywordsintag=true,
  morekeywords={alias,name,id},
  keywordstyle={\color{red}}
}


\lstset{style=psql}

\twocolumn[{

\vspace{1cm}

\begin{minipage}{0.5\linewidth}
  \colorbox{black}{\includegraphics[width=1\linewidth]{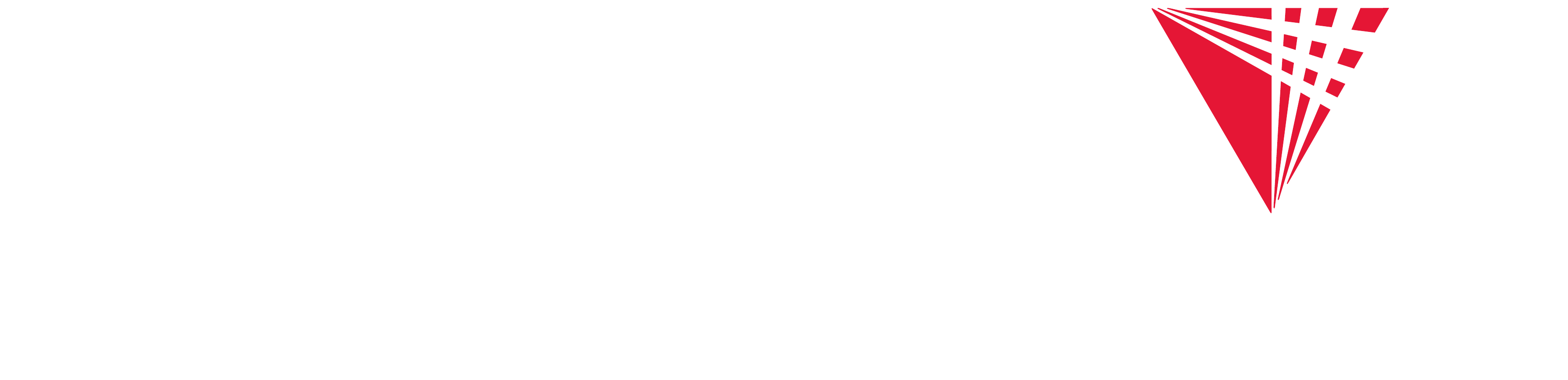}}
\end{minipage}
\hfill
\begin{minipage}{0.16\linewidth}
  \includegraphics[width=1\linewidth]{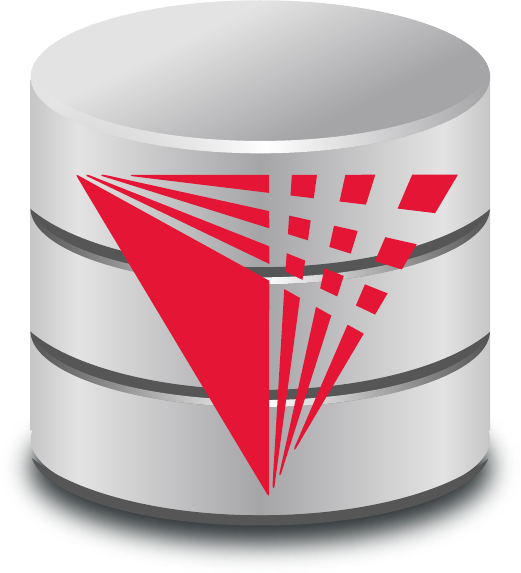}  
\end{minipage}\\
\vspace{4cm}

\centering
\begin{minipage}{1.0\linewidth} 
\centering
{\Huge \bf \TRtitle}
\end{minipage}
\\
\vspace{1cm}

{\huge \TRauthors}\\
\vspace{1cm}

{\huge \tt IIT DB Group Technical Report \TRnumber}\\
\vspace{1cm}

{\Large \TRdate}
\vspace{3cm}

{\huge \url{http://www.cs.iit.edu/~dbgroup/}}

\vspace{3cm}
\begin{minipage}{1.0\linewidth}
\textbf{LIMITED DISTRIBUTION NOTICE}: The research presented in this report may be submitted as a whole or in parts for publication  and will probably be copyrighted if accepted for publication. It has been issued as a Technical Report for early dissemination of its contents. In view of the transfer of copyright to the outside publisher, its distribution outside of IIT-DB prior to publication should be limited to peer communications and specific requests. After outside publication, requests should be filled only by reprints or legally obtained copies of the article (e.g. payment of royalties).  
\end{minipage}

}]

\clearpage


\begin{abstract}
Data provenance 
is essential for debugging query results, auditing data in cloud environments, and explaining outputs of Big Data analytics.  
%
A well-established technique is to represent provenance as annotations on data and
to \emph{instrument} queries to propagate these annotations to produce results annotated with provenance. 
However, even sophisticated optimizers are often incapable of producing efficient execution plans for instrumented queries, because of their inherent complexity and unusual structure.
Thus, while instrumentation enables provenance support for databases without requiring any modification to the DBMS, the performance of this approach is far from optimal.
In this work, we develop provenance-specific optimizations to address this problem.
Specifically, we introduce algebraic equivalences targeted at instrumented queries 
and discuss alternative, equivalent ways of instrumenting a query for provenance capture.
Furthermore, we present an extensible heuristic and cost-based optimization (CBO) framework that governs the application of these  optimizations and implement this framework in our \emph{GProM} provenance system. 
%
 Our CBO is agnostic to the plan space shape, uses a DBMS for cost estimation, and enables retrofitting of optimization choices into existing code by adding a few LOC. 
Our experiments confirm that these optimizations are highly effective, often improving performance by several orders of magnitude for diverse provenance tasks.
\end{abstract}



\section{Introduction}\label{sec:intro}
Database provenance, information about the origin of data and the queries and/or updates that produced it, is critical for debugging queries, auditing, establishing trust in data, and many other use cases.
%
%
%
The de facto standard for database provenance~\cite{KG12,karvounarakis2013collaborative,GA12} is to model provenance as annotations on data and define an annotated semantics for queries that determines how annotations propagate. 
Under such a semantics, each output tuple $t$ of a query $Q$ is annotated with its provenance, i.e.,
a combination  of input tuple annotations that explains how these inputs were used by $Q$ to derive $t$. 

\begin{figure}[t]
  \centering
\subfloat[
Provenance is captured using an annotated semantics of relational algebra which is compiled into standard relational algebra over a relational encoding of annotated relations and then translated into SQL code. 
]{  \label{fig:general-rewrite-approach}
  \begin{minipage}{0.96\linewidth}
  \centering
\includegraphics[width=1\columnwidth]{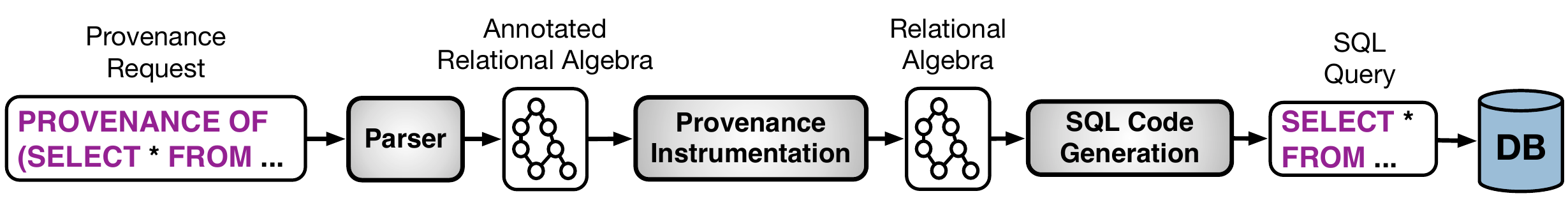}
\end{minipage}
}\\[-0.5mm]
\subfloat[
In addition to the steps of \textbf{(a)}, this pipeline contains a step called \emph{reenactment} that  compiles annotated updates into annotated query semantics. 
]{\label{fig:trans-rewrite-approach}
  \begin{minipage}{0.98\linewidth}
  \centering
  \includegraphics[width=0.9\columnwidth]{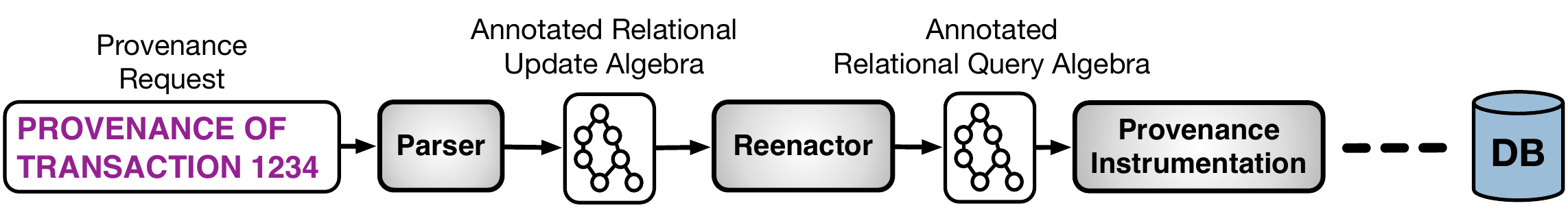}
\end{minipage}
}\\[-0.5mm]
\subfloat[Computing the edge relation of provenance graphs for Datalog queries based on a rewriting  called \emph{firing rules}. The instrumented Datalog program is compiled into relational algebra which in turn is translated into SQL.]{  \label{fig:DL-rewrite-approach}
\begin{minipage}{0.98\linewidth}
  \centering
\includegraphics[width=1\columnwidth]{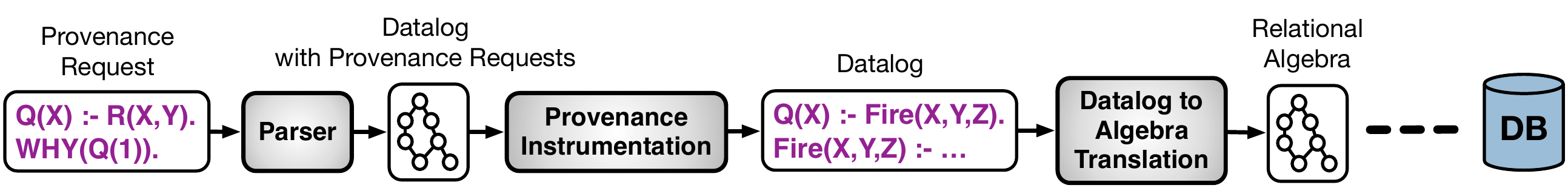}
\end{minipage}  
}
\caption{Instrumentation pipelines for capturing provenance for \textbf{(a)} SQL queries, \textbf{(b)} transactions, and \textbf{(c)} Datalog queries.}
%

\end{figure}



Database provenance systems such as  Perm~\cite{glavic2013using}, GProM~\cite{arab2014generic}, DBNotes~\cite{bhagwat2005annotation}, LogicBlox~\cite{GA12}, 
declarative Datalog debugging~\cite{KL12},   ExSPAN~\cite{ZS10}, and many others use a relational encoding of provenance annotations. These systems typically compile queries with annotated semantics into relational queries
that produce this encoding of provenance annotations following the process outlined in Fig.~\ref{fig:general-rewrite-approach}. We refer to this reduction from annotated to standard 
relational semantics as \textit{provenance instrumentation} or instrumentation for short.
The technique of compiling non-relational languages 
into relational languages (e.g., SQL) has also been 
applied for translating  XQuery into SQL over a shredded  representation of XML~\cite{grust2010let} and 
for compiling languages over nested collections into SQL~\cite{CL14}. 
%
%
The example below introduces a relational encoding of provenance polynomial~\cite{KG12} and the instrumentation approach for this model implemented in Perm~\cite{glavic2013using}.


\begin{Example}\label{ex:simple-prov-ex}
Consider a query  
over the database in Fig.~\ref{fig:Example-database} returning shops that sell items which cost more than \$20:
\\[-4mm]
$$
\projection_ {name} (shop \join_{name=shop} sale \join_{item=id} \selection_{price > 20}(item))
$$\\[-6mm]
The result of this query is shown in Fig.~\ref{fig:Example-database-result}. 
Using provenance polynomials
to represent provenance, tuples in the database are annotated with  variables
representing tuple identifier. We show these annotations to the left of each tuple. Each query result is annotated with
a polynomial over these variables that explains how the tuple was derived by combining input
tuples. The addition operation in these polynomials corresponds to alternative use of
tuples such as in a union or projection and multiplication represents
conjunctive use (e.g., a join). For example, the query result $(Walmart)$ was
derived by joining tuples $s_1$, $a_1$, and $i_1$ ($s_1 \cdot a_1 \cdot i_1$) or
alternatively by joining tuples $s_1$, $a_3$, and $i_3$
($s_1 \cdot a_3 \cdot i_3$).
Fig.~\ref{fig:provenance-result-example-database}
shows a relational encoding of these annotations as supported by the
Perm~\cite{glavic2013using} and GProM~\cite{arab2014generic} systems: 
variables are represented by the tuple
they are annotating, multiplication is represented by concatenating the encoding
of the factors, and addition is represented by encoding
each summand as a separate tuple.\footnote{The details 
  are
  beyond the scope of this paper, e.g., the input
  polynomial is normalized into a sum of products. The interested reader is
  referred to~\cite{glavic2013using}.}
This encoding is computed by compiling the input query with annotated semantics into  
 relational algebra. The resulting \emph{instrumented} query is shown below. This query adds
 attributes from the input relations to the final projection and renames them (represented as $\to$)
to denote that they store provenance.\\[-2mm]  
\resizebox{1\linewidth}{!}{
  \begin{minipage}{1.0\linewidth}
\begin{align*}
Q_{join} &= shop \join_{name=shop} sale \join_{item=id} \selection_{price > 20}(item)\\
Q&= \projection_{name,name \to P(name),numEmp \to P(numEmp), \ldots} (Q_{join})
\end{align*}
\end{minipage}
}\\[1mm]
The instrumentation we are using here is defined for any SPJ (Select-Project-Join) query (and beyond) based on a set of algebraic rewrite rules (see~\cite{glavic2013using} for details).
\end{Example}

\begin{figure}[t]
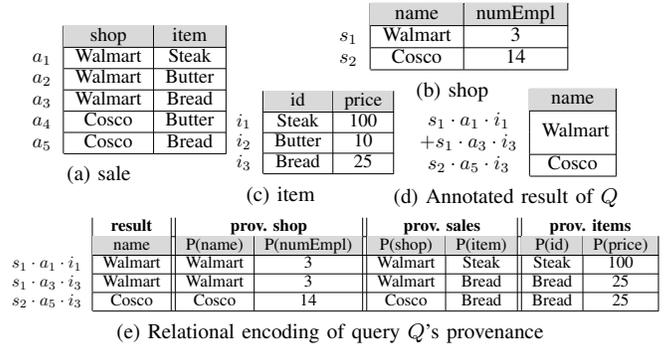

\centering
\begin{minipage}{0.30\linewidth}
  \centering


  \subfloat[sale]{%
    \centering
          \resizebox{0.8\linewidth}{!}{
            \begin{minipage}{1.0\linewidth}
              \centering
    \begin{tabular}{c|c|c|} \cline{2-3} 
     & \chead shop& \chead item\\ \cline{2-3}
     $a_1$ & Walmart & Steak\\ \cline{2-3}
     $a_2$ & Walmart & Butter\\ \cline{2-3}
     $a_3$ & Walmart & Bread\\ \cline{2-3}
     $a_4$ & Cosco & Butter \\ \cline{2-3} 
     $a_5$ & Cosco & Bread \\ \cline{2-3}
    \end{tabular}
  \end{minipage}
  }
  }
  
\end{minipage}
\begin{minipage}{0.65\linewidth}
  \centering
  \hspace{0.1\linewidth}
  \begin{minipage}{0.8\linewidth}
    \centering
    \subfloat[shop]{%
      \centering
      \resizebox{0.8\linewidth}{!}{
        \begin{minipage}{1.0\linewidth}
          \centering
    \begin{tabular}{c|c|c|} \cline{2-3}
 & \chead name& \chead numEmpl\\\cline{2-3}
   $s_1$ & Walmart & 3\\ \cline{2-3}
   $s_2$ & Cosco & 14 \\ \cline{2-3}
    \end{tabular}
  \end{minipage}
  }
  }
\end{minipage}\\[-2mm]

\begin{minipage}{0.33\linewidth}
  \centering
\subfloat[item]
{%
  \centering
            \resizebox{0.8\linewidth}{!}{
  \begin{minipage}{1.0\linewidth}
        \centering
    \begin{tabular}{c|c|c|} \cline{2-3} 
 & \chead id& \chead price\\ \cline{2-3}
     $i_1$ & Steak & 100\\ \cline{2-3} 
     $i_2$ & Butter & 10\\ \cline{2-3}
     $i_3$ & Bread & 25\\ \cline{2-3}
    \end{tabular}
  \end{minipage}
  }
  }
\end{minipage}\hspace{0.1cm}
\begin{minipage}{0.63\linewidth}
  \centering
  \subfloat[Annotated result of $Q$]{\label{fig:Example-database-result}
    \centering
    \resizebox{0.8\linewidth}{!}{
      \begin{minipage}{1.0\linewidth}
        \centering
    \begin{tabular}{c|c|} \cline{2-2}
 &\chead name\\ \cline{2-2} 
      $s_1 \cdot a_1 \cdot i_1 $  & \multirow{2}{1.1cm}{Walmart} \\
      $+ s_1 \cdot a_3 \cdot i_3$ &\\
      \cline{2-2}
  $s_2 \cdot a_5 \cdot i_3$& Cosco \\  \cline{2-2}
    \end{tabular}
  \end{minipage}
  }
  }
\end{minipage}
\end{minipage}\\[1mm]

\subfloat[Relational encoding of query $Q$'s provenance ]{\label{fig:provenance-result-example-database}
\resizebox{1\linewidth}{!}{
  \begin{tabular}{c|c||c|c||c|c||c|c|} 
&\multicolumn{1}{c||}{\bf result} & \multicolumn{2}{c||}{\bf prov. shop} & \multicolumn{2}{c||}{\bf prov. sales} & \multicolumn{2}{c}{\bf prov. items}\\ \cline{2-8} 
 & \chead  name& \chead P(name)& \chead P(numEmpl)& \chead P(shop)& \chead P(item)& \chead P(id)& \chead P(price)\\ \cline{2-8} 
 $s_1 \cdot a_1 \cdot i_1$  &  Walmart  & Walmart & 3 & Walmart & Steak & Steak & 100\\ \cline{2-8} 
$s_1 \cdot a_3 \cdot i_3 $&  Walmart  & Walmart & 3 & Walmart & Bread & Bread & 25\\ \cline{2-8} 
 $s_2 \cdot a_5 \cdot i_3$&  Cosco  & Cosco & 14 & Cosco & Bread & Bread & 25\\ \cline{2-8} 
  \end{tabular}
}
}\\[-2mm]

\caption{Example database and query provenance}
\label{fig:Example-database}

\end{figure}
\subsection{Instrumentation Pipelines}
\label{sec:instr-pipel}

\parttitle{Provenance for SQL Queries}
The instrumentation technique shown in Fig.~\ref{fig:general-rewrite-approach} and explained in the example above is applied by many relational provenance systems. For instance, the DBNotes~\cite{bhagwat2005annotation} system uses instrumentation to propagate attribute-level annotations according to Where-provenance~\cite{CC09}.
Variants of this particular instrumentation pipeline targeted in this work are discussed below.  While the query used for provenance computation in Ex.~\ref{ex:simple-prov-ex} is rather straightforward and is likely to be optimized in a similar fashion as the input query, this is not true for more complex provenance computations. 
%

\parttitle{Provenance for Transactions}
Fig.~\ref{fig:trans-rewrite-approach} shows a pipeline used to retroactively capture the provenance of updates and transactions~\cite{AG16a,AG17} in GProM. In addition to the steps from Fig.~\ref{fig:general-rewrite-approach}, this pipeline uses an additional compilation step called \textit{reenactment}. Reenactment translates transactional histories with annotated semantics into equivalent temporal queries with annotated semantics. Such queries can be executed using any DBMS with support for time travel to capture the provenance of a past transaction. While the details of this approach are beyond the scope of this work, consider the following simplified SQL example. The SQL update \lstinline!UPDATE R SET b = b + 2 WHERE a = 1! over relation \texttt{R(a,b)} can be reenacted using a query \lstinline!SELECT a, CASE WHEN a=1 THEN b+2 ELSE b END AS b! \lstinline!FROM R!. If executed over the version of the database seen by the update, this query is guaranteed to return the same result and have the provenance as the update.

\parttitle{Provenance for Datalog}
The pipeline shown in Fig.~\ref{fig:DL-rewrite-approach} generates provenance graphs that explain why or why-not a tuple is in the result of a Datalog query~\cite{LS16}. Such graphs store which successful and failed rule derivations of the query were relevant for (not) deriving the (missing) result tuple of interest. 
This pipeline compiles such provenance requests into a Datalog program that computes the edge relation of the provenance graph and then translates this program into  SQL. 

\parttitle{Provenance Export}
This pipeline extends Fig.~\ref{fig:general-rewrite-approach} with an additional step that translates the relational provenance encoding of a query produced by this pipeline into PROV-JSON,  which is the JSON serialization of the WC3 recommended provenance exchange format. This method~\cite{NX15} uses additional complex projections on top of the query instrumented for provenance capture to construct JSON document fragments and concatenate them into a single PROV-JSON document stored as a relation with a single attribute and single tuple.  
%




\subsection{Performance Bottlenecks of Instrumentation}
\label{sec:motivation}
While instrumentation enables diverse 
provenance features to be implemented on top of da\-ta\-bases without the need to modify the DBMS itself, the performance of generated queries is often far from optimal. Based on our extensive experience with instrumentation systems~\cite{LS16,NX15,arab2014generic,AG16a,glavic2013using} and a preliminary evaluation we have identified bad plan choices by the DBMS backend as a major   
bottleneck.
Since relational optimizers have to balance time spend on optimization versus improvement of query performance, optimizations that do not benefit common workloads are typically not considered.
Thus, most optimizers are incapable of simplifying instrumented queries, will not explore relevant parts of the plan space, or will spend excessive time on optimization. 
We now give a brief overview of problems that we have encountered in this space.
\parttitle{P1. Blow-up in Expression Size} The instrumentation for transaction provenance~\cite{AG16a,AG17} shown in Fig.~\ref{fig:trans-rewrite-approach} may produce queries with a large number of query blocks. This can lead to long optimization times in systems that unconditionally pull-up subqueries (such as Postgres) because the subquery pull-up would result in \lstinline!SELECT! clause expressions of size exponential in the number of stacked query blocks. While more advanced optimizers do not apply the subquery pull-up transformation unconditionally, they will at least consider it leading to the same blow-up in expression size during optimization. 

\parttitle{P2. Common Subexpressions} 
The Datalog provenance pipeline (Fig.~\ref{fig:DL-rewrite-approach}) instruments the input program using so-called firing rules to capture rule derivations. Compiling such queries into relational algebra leads to  algebra graphs with many common subexpressions and a large number of duplicate elimination operators. The provenance export instrumentation mentioned above constructs the PROV output using multiple projections over an instrumented subquery that captures provenance. The large number of common subexpressions in both cases may result in very long optimization time. Furthermore, if subexpressions are not reused then this significantly increases the query size. For the Datalog queries, the choice of when to remove duplicates significantly impacts performance. 

\parttitle{P3. Blocking Join Reordering} Provenance instrumentation as implemented in 
GProM~\cite{arab2014generic} is based on rewrite rules. 
For instance, provenance annotations are propagated through an aggregation by joining the aggregation with the provenance instrumented version of the aggregation's input on the group by attributes. Such transformations increase a query's size  and lead to 
interleaving of joins with operators such as aggregation or duplicate elimination. This interleaving may block optimizers from reordering joins  leading to suboptimal join orders. 

\parttitle{P4. Redundant Computations} Most provenance approaches instrument a query to capture provenance one operator at a time using operator-specific rewrite rules (e.g., the rewrite rules used by Perm~\cite{glavic2013using}). To be able to apply such operator-specific rules to rewrite a complex query, the rules have to be generic enough to be applicable no matter how operators are combined by the query. In some cases that may lead to redundant computations, e.g., an instrumented operator generates a new column that is not needed by any downstream operators. 

\begin{figure}[t]
\centering
\includegraphics[width=1\columnwidth]{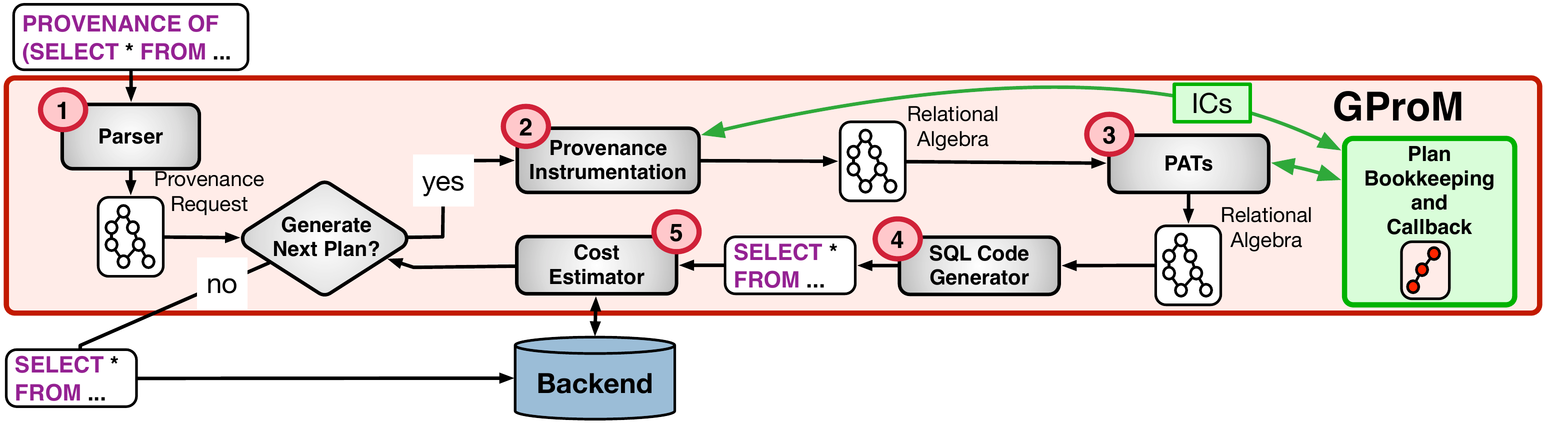}
$,$\\[-6mm]
\caption{GProM with Cost-based Optimizer}
\label{fig:cbo-arch}
\end{figure}


\section{Solution Overview}
\label{sec:solut-overv-contr}

We address the performance bottlenecks of instrumentation by developing heuristic and cost-based optimization techniques. 
%
While optimization has been recognized  as an important problem in provenance management, previous work has almost exclusively focused on how to compress provenance to reduce  storage cost, e.g., see~\cite{AB09,CJ08a,wu2013subzero}.
In contrast, in this work we assume that the provenance encoding is given, i.e., the user requests a particular type of provenance, 
and study the orthogonal problem of \textbf{improving the performance of instrumented queries} that compute provenance.


We now give a brief overview of our solution and 
 contributions.
An important advantage of our approach is that it applies to any database backend and instrumentation pipeline. 
New transformation rules and cost-based choices can be added with ease. 
We implement these optimizations in GProM~\cite{arab2014generic} (see Fig.~\ref{fig:cbo-arch}), our  provenance middleware that supports multiple DBMS backends
(available as open source at \url{https://github.com/IITDBGroup/gprom}).
%
%
Our optimizations which are applied during the compilation of a provenance request into SQL on average improve performance by over 4 orders of magnitude compared to unoptimized instrumented queries.
When optimizing instrumented queries, 
we can target any of the query languages used within the pipeline, 
e.g., 
if 
relational algebra is the output language for a compilation step then we can apply equivalence preserving transformations to the generated algebra expression before passing it on to the next stage of the pipeline. In fact, we develop several \textit{provenance-specific algebraic transformations} (or \textit{PAT}s for short). In addition, we can optimize during a compilation step, 
i.e., if we know two equivalent ways of translating an annotated algebra operator into standard relational algebra, we should choose the one which results in a better plan. 
We call such decisions \textit{instrumentation choices} (\textit{ICs}).
We developed
an effective set of PATs and ICs as 
our first major contribution.





\parttitle{PATs}
We identify algebraic equivalences which are usually not applied by databases, but are effective for speeding up provenance computations. For instance, 
we factor references to attributes to enable merging of projections without blow-off in expression size, pull up projections that 
create provenance annotations, and 
remove unnecessary duplicate elimination and window operators.
Following the approach presented in~\cite{grust2010let} we infer local and non-local properties such as candidate keys  for the algebra operators of a query. This enables us to define transformations that rely on non-local information.

\parttitle{ICs}
We introduce two ways for instrumenting an aggregation operator for provenance capture: 1) using a \textit{join} (this rule is used by Perm~\cite{glavic2013using}) to combine the aggregation with the provenance of the aggregation's input; 2) using \textit{window} functions (SQL \lstinline!OVER! clause) to directly compute the aggregation functions over inputs annotated with provenance. 
We also present two ways for pruning tuples that are not in the provenance early-on when computing the provenance of a transaction~\cite{AG16a}. 

\parttitle{CBO for Instrumentation}
Some  PATs are not always beneficial and for some ICs there is no clearly superior choice. Thus, there is a need for \textit{cost-based optimization} (CBO).
Our second contribution is 
a 
CBO framework for instrumentation pipelines. Our CBO algorithm can be applied to any such pipeline no matter what compilation steps and  intermediate languages are used. This is made possible by decoupling the plan space exploration from actual plan generation.

Our optimizer treats the instrumentation pipeline as a blackbox which it calls repeatedly  to  produce SQL queries (\textbf{plans}). Each such plan is sent to the backend database for 
planning and 
cost estimation. 
We refer to an execution of the pipeline as an \textbf{iteration}.
It is the responsibility of the pipeline's components to signal to the optimizer the existence of optimization choices (called \textbf{choice points}) through the optimizer's \textbf{callback API}. The optimizer responds to a call from one of these components by instructing it which of the available \textbf{options} to choose. 
We keep track of which choices had to be made, which options exist for each choice point, and which options were chosen. This information is sufficient to iteratively enumerate the plan space by making different choices during each iteration.
Our approach provides great flexibility in terms of supported optimization
decisions, e.g., we can choose whether to apply a PAT or select which ICs to use.
Adding a new optimization choice only requires adding a few LOC to the instrumentation pipeline to inform the optimizer about the availability of options. 
While our approach (Fig.~\ref{fig:cbo-arch}) has some aspects in common with cost-based query transformation~\cite{AL06}, it is to the best of our knowledge the first one that is \textbf{plan space and query language agonistic}.
Since costing a plan requires us to use the DBMS to optimize a query, the number of iterations that can be run 
within reasonable time is limited. In addition to randomized search techniques, we also support an approach that balances optimization vs. execution time, i.e., it stops optimization once a ``good enough'' plan has been found. 


Our approach peacefully coexists with the DBMS optimizer. We use the DBMS optimizer 
where it is effective (e.g., join reordering) and use our optimizer to address the database's shortcomings with respect to provenance computations. To maintain the advantage of database independence, we implement PATs and ICs in a middleware, but these optimizations could also be implemented as an extension of a regular database optimizer (e.g., as cost-based transformations~\cite{AL06}).

\section{Background and Notation}\label{sec:background}
%


A database schema ${\bf D} = \{{\bf R_1}, \ldots, {\bf R_n}\}$ is a set of relation schemas ${\bf R_1}$ to ${\bf R_n}$. A relation schema ${\bf R}(a_1, \ldots, a_n)$ consists of a name (${\bf R}$) and a list of attribute names $a_1$ to $a_n$. The arity of a relation schema is the number of attributes in the schema.
%
%
%
Here we use the bag-semantics version of the relational model.
%
 Let $\cal U$ be a domain of values. An instance $R$ of an n-ary relation schema ${\bf R}$  is a function ${\cal U}^n \to \mathbb{N}$ with finite support $\mid\{ t \mid R(t) \neq 0\}\mid$.  We use $t^m \in R$ to denote that tuple $t$ occurs with multiplicity $m$, i.e., $R(t) = m$ and $t \in R$ to denote that $R(t) > 0$.
%
An n-ary relation $R$ is contained in another n-ary relation $S$ iff $\forall t \in {\cal U}^n: R(t) \leq S(t)$, i.e., each tuple in $R$ appears in $S$ with the same or higher multiplicity. We abuse notation and write $R \subseteq S$ to denote that $R$ is contained in $S$. 



\begin{table}
  \centering
 \begin{tabular}{|p{1.3cm}|p{6cm}|} \hline 
\rowcolor[gray]{.9}  Operator & Definition\\ \hline 
  $\selection$ & $\selection _\theta (R) = \{ t^n|t^n \in R \wedge t \models \theta  \}$ \\ \hline
  $\projection$ & $\projection_A (R) = \{t^n|n = \sum_{u.A = t} R(u) \} $ \\ \hline
  $\union$ & $R \union S = \{ t^{n+m}|t^n \in R \wedge t^m \in S \} $\\ \hline
  $\intersection$ & $R \intersection S = \{ t^{min(n,m)}|t^n \in R \wedge t^m \in S \} $\\ \hline
  $\difference$ & $R-S = \{ t^{max(n-m,0)}|t^n \in R \wedge t^m \in S \}$ \\ \hline
  $\crossprod$ & $R \crossprod S = \{ (t,s)^{n*m} |t^n \in R \wedge s^m \in S \} $ \\ \hline
  $\aggregation$ & $_{G}\aggregation_{f(a)} (R) = \{ (t.G, f(G_t))^1|t \in R \} $ \\ 
                 & $G_t = \{ (t_1.a)^n |{t_1}^n \in R \wedge t_1.G = t.G \} $ \\ \hline
  $\duplicate$ & $\duplicate (R) = \{ t^{1} |t  \in R \} $ \\ \hline
$\omega$ & $\omega_{f(a) \to x, G\|O}(R) \equiv \{ (t,f(P_t))^n | t^n \in R \}  $\\ 
                 & $P_t = \{ (t_1.a)^n |t_1^n \in R \wedge t_1.G = t.G \wedge t_1 \leq_O t \} $ \\ \hline
 \end{tabular}\\[-2mm]
 \caption{Relational algebra operators}
\label{tab:rel-algebra-def}
\end{table}




Table~\ref{tab:rel-algebra-def} shows the definition of the bag-semantics version of relational algebra we use in this work.
In addition to set operators, selection, projection, crossproduct, duplicate elimination, and join, we also support aggregation and windowed aggregation.
Aggregation $_{G}\aggregation_{f(a)} (R)$ groups tuples according to their values in attributes $G$ and computes the aggregation function $f$ over the values of attribute $a$ for each group. 
Window operator $\omega_{f(a) \to x, G\|O}(R)$ applies function $f$ to the window generated by partitioning the input on expressions $G$ and ordering tuples by $O$. For each input tuple $t$, the window operator returns $t$ with an additional attribute $x$ storing the result of the window function. We use $Q(I)$ to denote the result of query $Q$ over database instance $I$.
We use 
$op_1 \ancestor op_2$ to denote that operator $op_2$ is an ancestor of (downstream of) $op_1$ and $op_1 \allAncestor op_2$ to denote that $op_2$ lies on all paths from $op_1$ the root of the query these operators belong to.
Furthermore, we use $Q[Q_1 \leftarrow Q_2]$ to denote the substitution of subexpression $Q_1$ in $Q$ with $Q_2$ and
 $\schema{Q}$ to denote the schema of the result of an algebra expression $Q$. 

\section{PAT}\label{sec:heuristic}

We now introduce equivalence-preserving PAT rules that address some of the problems mentioned in Sec.~\ref{sec:motivation}. 
These results are applied after the compilation of a provenance request into relational algebra. 
We present each rule as $\frac{pre}{q \rightarrow q'}$ which has to be read as ``If condition $pre$ holds, then $q$ can be rewritten as $q'$''. Similar to  Grust et al.~\cite{grust2010let}, we infer properties for the operators of an algebra expression and use these properties in preconditions of rules. 
These properties provide us with essential local  (e.g., candidate keys) and non-local information (e.g., which attributes of the operator's result are necessary for evaluating downstream operators).


\subsection{Operator Properties}
\parttitle{set}
Boolean property \emph{set} denotes 
if an ancestor of an operator $op$ is a duplicate elimination operator and the number of duplicates is irrelevant for computing operators on the path from $op$ to the next duplicate elimination operator. We use \emph{set} to remove or add duplicate elimination operators. 
\begin{Definition} \label{def:def_set}
  Let $op$ be an operator in a query $Q$. $set(op) = true$ iff $\exists op': op \allAncestor op'$ with $op' = \delta$ and $\forall op'': op \ancestor op'' \ancestor op'$ we have $op'' \not\in \{ \aggregation, \omega \}$.
\end{Definition}


\parttitle{keys} 
Keys are the candidate keys of an operator's output. For example, consider a relation $R(a,b,c,d)$ where $\{a\}$ and $\{b,c\}$ are unique, 
then  
$keys(R) = \{\{a\},\{b,c\}\}$. We compute candidate keys using a best effort approach since if we allow for user defined functions in projection and selection expressions it the problem of computing candidate keys for the relation produced by an operator is in undecidable. For instance, consider a projection $Q = \projection_{f(a) \to c, b}(R)$ over a relation $R(a_1, \ldots, a_n,b)$ where $f$ is expressed as a polynomial over $\{a_1, \ldots, a_n\}$. Furthermore, assume that $R$ has a single candidate key $\{a\}$. If $f$ is injective, then $Q$ has a candidate keys $\{c\}$, else $R$ either has no keys or a candidate key $\{c,b\}$. However, deciding whether the function $f$ is injective can be reduced to Hilbert's tenth problem which is known to be undecidable. Thus, it is not possible in general to compute the set of candidate keys for an operator's output. Using the same example, it is also easy to see that checking whether a set of attributes is a candidate key for the result of a query $Q$ is also undecidable.  Given this negative result, we restrict our approach to computing a set of keys that is not necessarily complete nor is each key in this set guaranteed to be minimal. Note that this is not a big drawback for our approach since our algebraic transformations will only reason over the fact that some key exists. That is, we may miss a chance of applying a transformation since our approach may not be able to determine a key that holds, but we will never apply a transformation that relies on the existence of a key if no such key exists.

\begin{Definition} \label{def:def_keys}
Let $Q = op(Q')$ be a query. A set 
$E \subseteq \schema{Q}$ is a \emph{super key} for $op$ iff for every instance $I$ we have  $\forall t,t' \in Q(I): t.E = t'.E \rightarrow t=t'$ and $t^n \in Q(I) \rightarrow n \leq 1$. 
\end{Definition}

\parttitle{EC}
The equivalence class (EC) property records equivalence classes. Equivalence classes are sets of attributes such that for an operator's output we are only interested in tuples where attributes in the same equivalence class  have the same value, i.e., this means we are allowed to enforce a restricted form of \textit{equality generating dependencies} (egds).\footnote{An egd is universal formula of the form $\forall \vec{x}: \phi(\vec{x}) \rightarrow \bigwedge_{i \in \{1,\ldots,m\}} x_{l_i} = x_{r_i}$ where $\vec{x} = x_1, \ldots, x_n$ is a vector of variables, $\phi$ is a conjunctive first-order formula over relational predicates and comparisons,  $m \leq n$, and for each $i$ we have $x_{l_i}, x_{r_i} \in \{1,\ldots, n\}$. Equivalence classes model edgs of the form $\forall \vec{x}: R(\vec{x}) \rightarrow x_i = x_j$ for $i,j \in \{1,\ldots, n\}$ or $\forall \vec{x}: R(\vec{x}) \rightarrow x_i = c$ where $c$ is a constant.} We model this condition using query equivalence: if $a$ and $b$ are in the same equivalence class for a query $Q$ then $Q \equiv \selection_{a=b}(Q)$. For example, if $EC(R)=\{\{a\},\{b,c\},\{d\}\}$ then we know that $t.b = t.c$ for each tuple $t \in R$. Note that in this example, the egd $\forall x_1,x_2,x_3,x_4: R(x_1, x_2,x_3,x_3) \rightarrow x_2 = x_3$ is guaranteed to hold and thus can be safely enforced (e.g., by adding a selection on $b=c$). There are other cases where we can enforce an egd even though this egd does not hold for the output of a particular operator. For instance, consider the query $\selection_{a=b}(R)$. Condition $a=b$ does not necessarily hold for all tuples in $R$, but it is obvious that we can enforce this condition.


\begin{Definition} \label{def:def_ec}
Let $Q_{sub} = op(Q_{sub}')$ be a subquery of query $Q$. A set of attributes $E \subseteq \schema{Q}$ is an \emph{equivalence class} (EC) for $op$ iff for all instances $I$ we have $\forall t \in Q(I):  \forall a,b \in E: Q \equiv Q[Q_{sub} \leftarrow \selection_{a=b}(Q_{sub})]$. Equivalence classes may also contain a constant $c$. In this case we additionally require that $\forall t \in Q(I):  \forall a \in E: Q \equiv Q[Q_{sub} \leftarrow \selection_{a=c}(Q_{sub})]$.
\end{Definition}

As is the case for candidate keys, we cannot hope to find an algorithm that computes all equivalences that can be enforced for any query using generalized projection with turing complete functions. Consider the query $Q = \projection_{f(a) \to b, g(a) \to c}(R)$ over relation $R(a)$ where $f$ and $g$ are functions expressed in a turing-complete language. If $f$ and $g$ compute the same function, then $b$ and $c$ will have the same value in every tuple in the result of $Q$. However, checking whether this is the case can be reduced to program equivalence which is an undecidable problem. We use an inference algorithm for equivalence classes that is sound, but not complete. That is, if the algorithm places attributes $a$ and $b$ in the same equivalence class then $a=b$ can safely be enforced, but there is no guarantee that the returned equivalence classes are maximal.


\parttitle{icols}
This property records which attributes are needed to evaluate ancestors of an operator. For example, attribute $d$ in $\projection_{a} (\projection_{a, b+c \to d}(R))$ is not needed to evaluate $\projection_a$.

\begin{Definition} \label{def:def_icols}
Let $Q$ be a query and $Q_{sub} = op(Q_{sub}')$ be a subquery of $Q$, $icols(op)$ is the minimal set of attributes $E \subseteq \schema{Q_{sub}}$ such that $Q \equiv Q[Q_{sub} \leftarrow \projection_{icols(op)}(Q_{sub})]$.
\end{Definition}



\subsection{Property Inference}
We infer properties for operators 
through 
traversals of the algebra graph.
During a bottom-up traversal the property $P$ for an operator $op$ is computed based on the values of $P$ for the operator's children. Conversely, during a top-down traversal the property $P$ of an operator $op$ is computed based on the values of $P$ for the  parents of $op$. 
We use $\circledast$ to denote the root of a query tree, for example, if $Q = \projection_{A,B}(R)$, then we enclose the top most operator using $\circledast$ (like $\circledast(\projection_{A,B}(R))$) to denote that this operator is the root of the query tree. 
In the following, we show the inference rules of property $EC$ in Table~\ref{tab:bottom-up} and~\ref{tab:top-down}, the inference rules of property $icols$ in Table~\ref{tab:top-down-icols}, the inference rules of property $set$ in Table~\ref{tab:top-down-set} and the inference rules of property $key$ in Table~\ref{tab:bottom-up-key}.

\begin{table*}

  \begin{minipage}{1.0\linewidth}
\centering
\caption{Bottom-up inference of property \textit{EC} (\textit{Equivalence Class}) of operator $\Diamond$}
\label{tab:bottom-up}
  \begin{tabular}{|c|c|c|} \hline 
\rowcolor[gray]{.9}  Rule & Operator $\Diamond$ & Inferred property \textit{EC} of $\Diamond$\\ \hline 
  1 & R & $\{\{a\}\mid a \in \schema{R} \}$ \\ \hline
  2 & $\selection_ {(\theta_{1} \wedge ... \wedge \theta_{n})}(R)$ & ${\cal E}^* (EC(R) \cup 
\{\{a,b\}\mid \exists i: \theta_{i}=(a=b) \} 
)$

    \\ \hline                                                          
 3 & $\projection_ {a_{1} \rightarrow b_{1},...,a_{n} \rightarrow b_{n}}(R)$ & ${\cal E}^* ( \{\{ b_i, b_j \} \mid \exists E \in EC(R) \wedge a_i \in E \wedge a_j \in E \} \union \{\{b_i\} | i \in \{1,\ldots,n\}\})$ \\
  
  \hline
 4 & $R \join_ {a=b} S$ & ${\cal E}^*(EC(R) \cup EC(S) \cup \{\{a,b\}\}) $\\ \hline
 5 & $R \crossprod S$ & $EC(R) \cup EC(S) $\\ \hline
 6 & $ _{b_{1},...,b_{n}} \aggregation _{F(a)}(R) $ & $\{  \{ b_{1},...,b_{n} \} \cap E \mid E \in EC(R) \}  \cup \{\{F(a)\}\} $ \\ \hline
 7 & $\duplicate(R)$ & $EC(R) $ \\ \hline
 8 & $R \union S$ & ${\cal E}^* (\{ E \cap E' \mid E \in EC(R) \wedge E' \in EC(S)[\schema{S}/\schema{R}] \}$
    \\ \hline
 9 & $R \intersection S$ & ${\cal E}^*( EC(R) \union EC(S)[\schema{S}/\schema{R}])$ \\ \hline
10 &  $R \difference S$ & $EC(R)$ \\ \hline
11 & $\omega_{f(a) \to x, G\|O}(R)$ & $EC(R) \cup \{\{x\}\}$ \\ \hline
  \end{tabular}
\end{minipage}
\begin{minipage}{1.0\linewidth}
  \centering
\caption{Top-down inference of property \textit{EC} (\textit{Equivalence Class}) for the input(s) of operator $\Diamond$}
\label{tab:top-down}
  \begin{tabular}{|c|c|c|} \hline 
\rowcolor[gray]{.9}  Rule & Operator $\Diamond$ & Inferred property \textit{EC} of input(s) of $\Diamond$\\ \hline 
  1 & $\selection_ {(\theta_{1} \wedge ... \wedge \theta_{n})}(R)$ & $ EC(R) = {\cal E}^* (EC(\selection_ {(\theta_{i} \wedge ... \wedge \theta_{n})}(R)) \union EC(R))$ 
  
  \\ \hline
  2 & $\projection_ {a_{1} \rightarrow b_{1},...,a_{n} \rightarrow b_{n}}(R)$ & $EC(R) = {\cal E}^* ( \{\{ a_i, a_j \} \mid \exists E \in EC(\projection_ {a_{1} \rightarrow b_{1},...,a_{n} \rightarrow b_{n}}(R)) \wedge b_i \in E \wedge b_j \in E \} \union EC(R))$ \\ \hline
3 & $R \join_ {a=b} S$ & $EC(R) = {\cal E}^* ( \{E-\schema{S}|E \in EC(R \join_ {a=b} S)\} \union EC(R)) $   \\ 
                   &  & $EC(S) = {\cal E}^* ( \{E-\schema{R}|E \in EC(R \join_ {a=b} S)\} \union EC(S))$ \\ \hline
                     
 4 & $R \crossprod S$ & $EC(R) = {\cal E}^*( \{E-\schema{S}|E \in EC(R \join_ {a=b} S)\} \union EC(R)) $   \\ 
                &   & $EC(S) = {\cal E}^*(\{E-\schema{R}|E \in EC(R \join_ {a=b} S)\} \union EC(S)) $ \\ \hline
5 &  $_{b_{1},...,b_{n}} \aggregation _{F(a)}(R)$ & $EC(R) = {\cal E}^*(\{ E \cap \schema{R} |E \in EC(_{b_{1},...,b_{n}} \aggregation _{F(a)}(R))\} \cup EC(R))  $ \\ 
   \hline
 6 & $\duplicate(R)$ & $EC(R) = {\cal E}^*(EC(\duplicate(R)) \union EC(R))$ \\ \hline
  7 & $R \union S$ & $ EC(R) = {\cal E}^*(EC(R \union S) \cup EC(R))$ \\ 
                &                                   & $ EC(S) = {\cal E}^*(EC(R \union S)[\schema{R}/\schema{S}] \cup EC(S))$ \\ \hline
8 &  $R \intersection S$ & $ EC(R) = {\cal E}^* (EC(R \intersection S) \union EC(R))$  \\ 
                                               &           & $ EC(S) = {\cal E}^* (EC(R \intersection S)[\schema{R}/\schema{S}] \union EC(S))$\\ \hline
 9 & $R \difference S$ & $ EC(R) = {\cal E}^* (EC(R \difference S) \union EC(R))$ \\ 
                                   &                     & $ EC(S) = \{\{a\}\mid a \in \schema{S} \}$ \\ \hline
10 & $\omega_{f(a) \to x, G\|O}(R)$ & $ EC(R) = {\cal E}^*(\{ E \cap \schema{R} |E \in EC(\omega_{f(a) \to x, G\|O}(R))\} \cup EC(R))$  \\ \hline
  \end{tabular}
\end{minipage}
\begin{minipage}{1.0\linewidth}
  \centering
\caption{Top-down inference of property \textit{EC} (\textit{Equivalence Class}) for the input(s) of operator $\Diamond$}
\label{tab:top-down}
  \begin{tabular}{|c|c|c|} \hline 
\rowcolor[gray]{.9}  Rule & Operator $\Diamond$ & Inferred property \textit{EC} of input(s) of $\Diamond$\\ \hline 
  1 & $\selection_ {(\theta_{1} \wedge ... \wedge \theta_{n})}(R)$ & $ EC(R) = {\cal E}^* (EC(\selection_ {(\theta_{i} \wedge ... \wedge \theta_{n})}(R)) \union EC(R))$ 
  
  \\ \hline
  2 & $\projection_ {a_{1} \rightarrow b_{1},...,a_{n} \rightarrow b_{n}}(R)$ & $EC(R) = {\cal E}^* ( \{\{ a_i, a_j \} \mid \exists E \in EC(\projection_ {a_{1} \rightarrow b_{1},...,a_{n} \rightarrow b_{n}}(R)) \wedge b_i \in E \wedge b_j \in E \} \union EC(R))$ \\ \hline
3 & $R \join_ {a=b} S$ & $EC(R) = {\cal E}^* ( \{E-\schema{S}|E \in EC(R \join_ {a=b} S)\} \union EC(R)) $   \\ 
                   &  & $EC(S) = {\cal E}^* ( \{E-\schema{R}|E \in EC(R \join_ {a=b} S)\} \union EC(S))$ \\ \hline
                     
 4 & $R \crossprod S$ & $EC(R) = {\cal E}^* \{E-\schema{S}|E \in EC(R \join_ {a=b} S)\} \union EC(R)) $   \\ 
                &   & $EC(S) = \{E-\schema{R}|E \in EC(R \join_ {a=b} S)\} \union EC(S)) $ \\ \hline
5 &  $_{b_{1},...,b_{n}} \aggregation _{F(a)}(R)$ & $EC(R) = {\cal E}^*(\{ E \cap \schema{R} |E \in EC(_{b_{1},...,b_{n}} \aggregation _{F(a)}(R))\} \cup EC(R))  $ \\ 
   \hline
 6 & $\duplicate(R)$ & $EC(R) = {\cal E}^*(EC(\duplicate(R)) \union EC(R))$ \\ \hline
  7 & $R \union S$ & $ EC(R) = {\cal E}^*(EC(R \union S) \cup EC(R))$ \\ 
                &                                   & $ EC(S) = {\cal E}^*(EC(R \union S)[\schema{R}/\schema{S}] \cup EC(S))$ \\ \hline
8 &  $R \intersection S$ & $ EC(R) = {\cal E}^* (EC(R \intersection S) \union EC(R))$  \\ 
                                               &           & $ EC(S) = {\cal E}^* (EC(R \intersection S)[\schema{R}/\schema{S}] \union EC(S))$\\ \hline
 9 & $R \difference S$ & $ EC(R) = {\cal E}^* (EC(R \difference S) \union EC(R))$ \\ 
                                   &                     & $ EC(S) = \{\{a\}\mid a \in \schema{S} \}$ \\ \hline
10 & $\omega_{f(a) \to x, G\|O}(R)$ & $ EC(R) = {\cal E}^*(\{ E \cap \schema{R} |E \in EC(\omega_{f(a) \to x, G\|O}(R))\} \cup EC(R))$  \\ \hline
  \end{tabular}
\end{minipage}
\begin{minipage}{0.5\linewidth}
  \centering
\caption{Top-down inference of property \textit{icols} for the input(s) of operator $\Diamond$}
  \resizebox{0.83\linewidth}{!}{
  \begin{minipage}{0.9\linewidth}
\centering

\label{tab:top-down-icols}
  \begin{tabular}{|c|c|c|} \hline 
\rowcolor[gray]{.9} Rule & Operator $\Diamond$ & Inferred property \textit{icols} of input(s) of $\Diamond$\\ \hline 
 1 & $ \circledast(R)$ & $ icols(R) =  \schema{R}$ 
  
  \\ \hline
  
  2 & $\selection_ {\theta}(R)$ & $  icols(R)  =icols \union cols(\theta) $   
  
  \\ \hline
3 & $\projection_ {A_{1} \rightarrow B_{1},...,A_{n} \rightarrow B_{n}(R)}$ & $   icols(R)  = \displaystyle\bigcup_{\forall i:i \in \{1,...,n\} \wedge B_i \in icols}^{} cols(A_i)$ \\ \hline
 4 & $R \join_ {a=b} S$ & $  icols(R)  = (icols \union \{a,b\}) \intersection \schema{R}$   \\ 
                  &   & $  icols(S)  = (icols \union \{a,b\}) \intersection \schema{S}$   \\ \hline
                     
  5 & $R \crossprod S$ & $   icols(R)  = icols \intersection \schema{R} $   \\ 
                 &  & $  icols(S)  = icols \intersection \schema{S} $ \\ \hline
 6 & $_{b_{1},...,b_{n}} \aggregation _{F(a)}(R)$ & $ icols(R) = \{b_1, ... , b_n, a \}$
\\   \hline
 7 & $\duplicate(R)$ & $  icols(R)  = \schema{R} $ \\ \hline
8 &  $R \union S$ & $  icols(R)  = icols $ \\ 
                                                 &  & $  icols(S)  = icols [\schema{R}/\schema{S}] $ \\ \hline
 9 & $R \intersection S$ & $  icols(R)  = \schema{R}  $  \\ 
                                                     &     & $  icols(S)  = \schema{S}  $\\ \hline
 10 & $R \difference S$ & $  icols(R)  = \schema{R}  $ \\ 
                                                     &   & $  icols(S)  = \schema{S} $ \\ \hline
11 & $\omega_{f(a) \to x, G\|O}(R)$ &  $icols(R) = icols - \{x\} \cup \{a\} \cup G \cup O$\\ \hline
  \end{tabular}
 \end{minipage}
 }
\end{minipage}
\begin{minipage}{0.5\linewidth}
\centering
\caption{Top-down inference of Boolean property \textit{set} for operator $\Diamond$}
\label{tab:top-down-set}
  \begin{tabular}{|c|c|} \hline 
\rowcolor[gray]{.9}  Operator $\Diamond$ & Inferred property \textit{set} of input(s) of $\Diamond$\\ \hline 
    $\circledast(R)$ & $ set(R) = false$ 
  
  \\ \hline
  
  $\selection_ {(\theta_{1} \wedge ... \wedge \theta_{n})}(R)$ & $ set(R) = set(R) \wedge set $ 
  
  \\ \hline
  $\projection_ {a_{1} \rightarrow b_{1},...,a_{n} \rightarrow b_{n}(R)}$ & $  set(R) = set(R) \wedge set  $ \\ \hline
 $R \join_ {a=b} S$ & $  set(R) = set(R) \wedge set  $   \\ 
                     & $  set(S)= set(S) \wedge set  $ \\ \hline
                     
  $R \crossprod S$ & $  set(R) = set(R) \wedge set  $   \\ 
                   & $  set(S) = set(S) \wedge set  $ \\ \hline
  $_{b_{1},...,b_{n}} \aggregation _{F(a)}(R)$ & $ set(R) =  false  $ \\ 
   \hline
  $\duplicate(R)$ & $  set(R) = set(R) \wedge true $ \\ \hline
  $R \union S$ & $ set(R) = set(R) \wedge set $ \\ 
                                                   & $ set(S) = set(S) \wedge set  $ \\ \hline
  $R \intersection S$ & $ set(R) = set(R) \wedge set  $  \\ 
                                                          & $set(S) = set(S) \wedge set   $\\ \hline
  $R \difference S$ & $ set(R) = set(R) \wedge set $ \\ 
                                                        & $set(S) = set(S) \wedge set  $ \\ \hline
$\omega_{f(a) \to x, G\|O}(R)$ & $set(R) = false$ \\ \hline
  \end{tabular}  
\end{minipage}

\end{table*}

\begin{table*}[h]
\centering
\caption{Bottom-up inference of property \textit{key} for operator $\Diamond$}
\label{tab:bottom-up-key}
  \begin{tabular}{|c|c|c|} \hline 
\rowcolor[gray]{.9}  Rule & Operator $\Diamond$ & Inferred property \textit{key} of $\Diamond$\\ \hline 
  1 & $\selection_ {\theta}(R)$ & $ key = key(R) $ 
  
  \\ \hline
     2 & $\projection_ {a_{1} \rightarrow b_{1},...,a_{n} \rightarrow b_{n}(R)}$ 
                     & $ key = \{ E[B/A] | E \in key(R) \wedge E \subseteq \{a_1,..., a_n\} \}$ for $A = \{a_1,\ldots, a_n\}$ and $B = \{b_1, \ldots, b_n\}$ \\ \hline
    3 & $R \join_ {a=b} S$ &   $key =
MIN( \{ (E_1 \cup \{a\}) \cup (E_2 - \{b\}) | E_1 \in key(R) \wedge E_2 \in key(S) \}$\\
&&\hspace{15mm}$\cup  \{ (E_2 \cup \{b\}) \cup (E_1 - \{a\}) | E_1 \in key(R) \wedge E_2 \in key(S) \})$
  \\ \hline
  4 & $R \crossprod S$ & $ key = \{E_1 \union E_2 | E_1 \in key(R) \wedge E_2 \in key(S)\} $   \\ \hline
                   

    5 & $_{b_{1},...,b_{n}} \aggregation _{F(a)}(R)$  & $key =
                                                        \begin{cases}
                                                          \{\{b_1,..., b_n\}\}  & \mathtext{if} \nexists x \in key(R): x \subseteq  \{b_1,..., b_n\}\\
                                                          \{ \{b_1,..., b_n\}  \intersection x \mid x \in key(R) \wedge x \subseteq  \{b_1,..., b_n\} \} &\mathtext{otherwise}\\
                                                        \end{cases}$ \\
   \hline
6 & $\aggregation_{F(a)}(R)$ & $\{\{F(a)\}\}$ \\ \hline
    7 & $\duplicate(R)$ & $key =
                          \begin{cases}
                            key(R) & \mathtext{if} key(R) \neq \emptyset\\
                            \{\schema{R}\} & \mathtext{otherwise}\\
                          \end{cases}$ \\
  \hline
  8 & $R \union S$ & $ key = \emptyset $ \\ \hline 
                                                    
  9 & $R \intersection S$ & $ key = key(R) \union key(S)[\schema{R}/\schema{S}] $  \\ \hline
                                                          
  10 & $R \difference S$ & $ key = key(R)  $ \\ \hline

  11 & $\omega_{f(a) \to x, G\|O}(R)$ & $ key = key(R)  $ \\ \hline
  \end{tabular}
\end{table*}

\parttitle{Inferring the icols Property}
We compute icols in a top-down traversal, Table~\ref{tab:top-down-icols} shows the inference rules for the top-down traversal. We initialize the virtual root $\circledast$ of the algebra tree for query $Q$ with the attributes of $\schema{Q}$ since all these attributes are returned as part of the query result (Rule 1).  All attributes from a selection's condition $\theta$ (denoted by $cols(\theta)$ are needed to evaluate the selection. Thus, all attributes needed to evaluate the ancestors plus $cols(\theta)$ are required to evaluate the selection (Rule 2). For a projection we need all attributes that are used to compute result attributes of the projection that are part of icols (Rule 3). For crossproduct we restrict the columns needed to evaluate ancestors of the cross product to its inputs (Rule 5). The inference rule for join is a combination of the rules for cross product and join (Rule 4). For an aggregation we need all its group-by attributes to guarantee that the same number of tuples are returned (even if some group-by attributes are not needed by ancestors). Additionally, we need the attribute used as input to the aggregation function (Rule 6). All input attributes (Rule 7) are needed to evaluate a duplicate elimination operator (removing an attribute may change the number of results). The number of duplicates produced by a bag union is not affected by additional projections (Rule 8). Thus, only attributes needed to evaluate ancestors of the union are needed (her $\schema{R}/\schema{S}$ denotes renaming the attributes from $\schema{R}$ as $\schema{S}$. For both intersection and difference we need all input attributes to not affect the result (Rules 9 and 10). To evaluate a window operator we need all attributes that are used to compute the aggregation function, order-by and partition-by parameters (Rule 11).

\parttitle{Inferring the set Property}
We compute set in a top-down traversal, Table~\ref{tab:top-down-set} shows the inference rules for the top-down traversal. Note that the definition of this property is purely syntactical. Thus, the inference rules simply implement the condition of this property in a recursive fashion. For instance, the child of any aggregation or window operator  does not fullfil this condition since then the aggregation (respective window) operator will be on a path between the operator and the ancestor duplicate elimination operator if it exists.

\parttitle{Inferring the key Property}
We compute key in a bottom-up traversal, Table~\ref{tab:bottom-up-key} shows the inference rules for the bottom-up traversal. In the inference rules we define $MIN(S) = \{e| e\in S \wedge \nexists e' \in S : e' \subset e\}$ to remove any key that is a superset of another key, e.g., $\{a,b,c\}$ contains $\{a,b\}$. Any key that holds for the input of a selection is naturally also a key of the selection's output since a selection returns a subset of its input relation (Rule 1). A projection returns one result tuple for every input tuple. 
Thus, any key $k$ that holds over the input of a projection operator will hold (modulo renaming) in the projection's output unless some of the key attributes are projected out (Rule 2). 
A cross product returns all combinations $(t,s)$ of tuples $t$ from the left and $s$ from the right  input. It is possible to uniquely identify $t$ and $s$ using a pair of keys from the left and right input (Rule 4). For all tuples returned by an equality join $R \join_{a=b} S$, we know that $a=b$ holds. Thus, the functional dependencies $a \to b$ and $b \to a$. 
Thus, attribute $a$ can be substituted with attribute $b$ in any key for $R$ and $b$ with $a$ in any key for $S$ without changing the values of the key attributes for any tuple. For any key $k$ in $key(R)$ we know that $k \to a$ and, thus, $k \to b$. By a symmetric argument for any $k$ in $key(S)$ we have $k \to a$. It follows, that for $k \in key(R)$ and $k' \in key(S)$, $k \cup k' - \{b\}$ and $k \cup k' - \{a\}$ are keys for the join. Since for two keys $k$ and $k'$ generated in this fashion, it may be the case that $k \subseteq k'$, we apply $MIN$ to remove keys that contain other keys (Rule 3). For aggregation operators we consider two cases: 1) aggregation with group-by and 2) without group-by. For an aggregation with group-by the values for group-by attributes are unique in the output and, thus, are a superkey for the relation. Furthermore, all keys that are subsets of the group-by attributes are still keys in the output. Hence, if none of the keys are contained in the group-by attributes we can use the group-by attributes as a key and otherwise use all contained keys (Rules 5). Aggregation without group-by returns a single tuple. For this type of aggregation the aggregation function result is a trivial key (Rule 6). The bag union of two input relations does not have a key even if both inputs have keys because we do not know whether they values for these keys overlap (Rule 8). The result relation computed by an intersection $R \intersection S$ is a subset of both $R$ and $S$ (Rule 9). Thus, any key from either input is guaranteed to hold over the output (of course attributes from keys of $S$ have to be renamed). Set difference returns a subset of the left input relation. Thus, any key that holds over the left input is guaranteed to hold over the output (Rules 10). The window operator adds a new attribute value to every tuple from its input. Thus, every key that holds over the input also holds over the window operator's output (Rule 11).

\parttitle{Inferring the EC Property}
%
We compute EC in a bottom-up traversal followed by a top-down traversal.
Table~\ref{tab:bottom-up} 
shows the inference rules for the bottom-up traversal. 
In the inference rules we use an operator ${\cal E}^*$ that takes a set of ECs as input and merges classes if they are overlapping. This corresponds to repeated application of transitivity: $a = b \wedge b = c \rightarrow a=c$.
Formally, operator ${\cal E}^*$  is defined as the least fixed-point of operator $\cal E$ shown below:\\[-12mm]
%
\begin{center}
\resizebox{1\linewidth}{!}{
  \begin{minipage}{1.1\linewidth}
    \begin{align*}
      {\cal E}(EC) = &\{ E \union E' \mid E \in EC \wedge E' \in EC \wedge E \cap E' \neq \emptyset \wedge E \neq E' \} \\
                     &\cup \{ E \mid E \in EC \wedge  \not\exists
                       E' \in EC: E \cap E' \neq \emptyset \}
    \end{align*}
  \end{minipage}
}
\end{center}

In the following we first discuss the bottom-up inference rules and then the top-down rules.

\parttitle{Top-down Inference of EC}
For a base relation we place each attribute in its own equivalence class.
For selections (Rule 2), we assume that selection conditions have been translated into  conjunctive normal form $\theta_{1} \wedge ... \wedge \theta_{n}$. 
If $\theta_i$ is an equality comparison $a=b$ where both $a$ and $b$ are attributes of relation R, then we need to merge their ECs
(we denote the EC containing attribute $a$ as $EC(R,a)$).
If $b$ is a constant, then we need to add $b$ to $EC(R,a)$.  
This is realized by adding $\{a,b\}$ to the input set of ECs and merging overlapping ECs using ${\cal E}^*$.
For example, if  ${\bf R} = (a,b,c)$ and $EC(R)=\{\{a,b\},\{c\}\}$, then $EC(\selection_{a=5 \wedge c<9}(R)) = \{\{a,b,5\},\{c\}\}$.             
Any equivalences $a=b$ that hold over the input of a projection (Rule 2), also hold over its output as long as attributes $a$ and $b$ are present in the output. Since only a subset of the attributes from an equivalence class of the input may be present, we reconstruct equivalence classes based on present attributes using the ${\cal E}^*$ operator. 
\textit{Join}: For each tuple $t$ in the result of a join $R \join_ {a=b} S$ 
, we know that $t.a = t.b$, because otherwise the tuple cannot be in the result. Thus, we merge $EC(R,a)$  with $EC(S,b)$ using the approach presented above (adding $\{a,b\}$).
For example, if 
$EC(R)=\{\{a,b\},\{c\}\}$ and 
$EC(S)=\{\{d\},\{e,f\}\}$. Then, $EC(R \join_{a=d} S)=\{\{a,b,d\},\{c\},\{e,f\}\}$.     
A crossproduct (Rules 5) does not enforce any

\parttitle{Union} 
In the union rule $R \union S$ (Table~\ref{tab:bottom-up}), we rename the attributes of $S$ in $EC(S)$ to the attribute of $R$ by using $EC(S)[\schema{S}/\schema{R}]$. Then we combine $EC(S)$ with $EC(R)$. Finally, we use ${\cal E}^*$ to merge  overlapping classes. For example, if relation $R$ has schema $\{A,B\}$ with $EC=\{\{A\},\{B\}\}$ and relation $S$ has schema $\{C,D\}$ with $EC=\{\{C,D\}\}$, then $EC(R \union S)=\{\{A\},\{B\}\}$.

\begin{figure*}[t]
  \centering
%
\begin{minipage}{0.46\linewidth}
\begin{equation}\label{eq:pulling-up-provenance-projections}
     \frac{ a \subseteq \schema{\Diamond(\projection_{A}(R))} \wedge b \not\in icols(\Diamond(\projection_{A}(R))) }{\Diamond(\projection_{A,a \to b}(R)) \to  
     \projection_{\schema{\Diamond(\projection_{A}(R))},a \to b}(\Diamond(\projection_A (R)))}
\end{equation}
\end{minipage}
\hspace{1cm}
\begin{minipage}{0.15\linewidth}  
\begin{equation}\label{eq:duplicate-remove}
    \frac{keys(R) \neq \emptyset}{\duplicate (R) \rightarrow R} 
\end{equation}
\end{minipage}
\hspace{1cm}
\begin{minipage}{0.15\linewidth}  
\begin{equation}\label{eq:duplicate-remove-set}
    \frac{set(\duplicate(R))}{\duplicate (R) \rightarrow R} 
\end{equation}
\end{minipage}\\[-1mm]
\begin{minipage}{0.16\linewidth}  
\begin{equation}\label{eq:remove-redundant-columns1}
 \frac{A=icols(R)}{R \rightarrow \projection_A (R)}
\end{equation}
\end{minipage}
\hspace{1cm}
\begin{minipage}{0.25\linewidth}  
\begin{equation}\label{eq:window-function}
 \frac{x \not\in icols(\omega_{f(a) \to x} (R))}{\omega_{f(a) \to x}(R) \rightarrow R}
\end{equation}
\end{minipage}
\hspace{1cm}
\begin{minipage}{0.45\linewidth}  
\begin{equation}\label{eq:attribute-factoring}
 \frac{e_1 = if \  \theta \  then \  A + c \  else \  A}
{\projection_{e_1,...,e_m}(R) \to \projection_{ A +  \  if \  \theta \  then \  c \  else \  0, e_2,...,e_m}(R)}
\end{equation}
\end{minipage}\\[-2mm]
  \caption{Provenance-specific transformation (PAT) rules}
  \label{fig:algebraic-rules}
\end{figure*}

\subsection{Provenance-specific Transformations Rules}
\label{subsection:PATs-rules}

We now introduce the subset of our PAT rules shown in Fig.~\ref{fig:algebraic-rules}, prove their correctness, and then discuss how these rules address the performance bottlenecks discussed in Sec.~\ref{sec:motivation}.

\parttitle{Provenance Projection Pull Up}
Provenance instrumentation~\cite{arab2014generic,glavic2013using} 
seeds provenance annotations by duplicating attributes of the input using projection. This increases the size of tuples in intermediate results. 
We can delay this duplication of attributes if the attribute we are replicating is still available in ancestors of the projection. 
In Rule~\eqref{eq:pulling-up-provenance-projections}, 
$b$ is an attribute storing provenance generated by duplicating attribute $a$. If $a$ is available in the schema of $\Diamond(\projection_{A}(R))$ ($\Diamond$ can be any operator) and $b$ is not needed to compute $\Diamond$,
then  we can pull the projection on $a \to b$ through operator $\Diamond$. 
For example, consider a query $Q = \selection_{a<5}(R)$ over relation $R(a,b)$.
Provenance instrumentation yields: 
$\selection_{a<5}(\projection_{a,b,a \to P(a), b \to P(b)}(R))$. This projection can be pulled up to reduce the size of the selection's result tuples:
$
\projection_{a,b,a \to P(a), b \to P(b)}(\selection_{a<5}(R))
$.


\parttitle{Remove Duplicate Elimination}
Rules~\eqref{eq:duplicate-remove} and~\eqref{eq:duplicate-remove-set} remove duplicate elimination operators. If a relation $R$  has at least one candidate key, then it cannot contain any duplicates. Thus, a duplicate elimination applied to  $R$ can be safely removed (Rule~\eqref{eq:duplicate-remove}). Furthermore, if the output of a duplicate elimination $op$ is again subjected to duplicate elimination further downstream and the operators on the path between these two operators are not sensitive to the number of duplicates (property \textit{set} is true), then $op$ can be removed (Rule~\eqref{eq:duplicate-remove-set}).



\parttitle{Remove Redundant Attributes}
Recall that $icols(R)$ are the attributes of relation $R$ which are needed to evaluate ancestors of $R$ in the query.
If $icols(R) = A$, 
 then we use Rule~\eqref{eq:remove-redundant-columns1} to remove all other attributes by projecting $R$ on $A$.
For example, if 
$icols(R)=\{a, b\}$, then $R \to \projection_{a,b}(R)$ is a valid rewrite.
Operator $\omega_{f(a) \to x} (R)$ extends each tuple $t \in R$ by adding a new attribute $x$ that stores the result of window function $f(a)$.
Rule~\eqref{eq:window-function} removes $\omega$ if $x$ is not needed by any ancestor of $\omega (R)$.  
For example, if $\schema{R} = \{a,b\}$, then  $\schema{\omega_{sum(b) \to x} (R)} = \{a,b,x\}$. 
If $icols(\omega_{sum(b) \to x} (R)) =\{a,b\}$, then we can remove the window operator. 

\parttitle{Attribute Factoring}\label{sec:PAT-factor-attrs}
Attribute factoring restructures projection expressions in such a way that adjacent projections can be merged without blow-up in expression size. 
For instance, when projections $\projection_{b+b+b \to c}(\projection_{a+a+a\to b}(R))$ are merged, this increases the number of references to $a$ to 9 (each mention of $b$ is replaced with $a+a+a$). This blow-up can occur when computing the provenance of transactions where multiple levels of \lstinline!CASE! expressions are used. In relational algebra we represent  \lstinline!CASE! as $if \ \theta \  then \ e_1 \ else \ e_2$.   Rule~\eqref{eq:attribute-factoring} addresses a common case that arises when reenacting an update~\cite{AG16a}. 
For example, update \lstinline!UPDATE R SET a = a + 2 WHERE b = 2! would be expressed as $\projection_{if \ b = 2 \ then \ a+2 \ else \ a, b}(R)$ which can be rewritten as $\projection_{a + if \ b = 2 \ then \ 2 \ else \ 0, b}(R)$. Note how the number of references to attribute $a$ was reduced from 2 to 1. We define analog rules for any arithmetic operation which has a neutral element (e.g.,  multiplication). 


\subsection{Addressing Instrumentation Bottlenecks through PATs}
\label{sec:rule-problem-address}

Rule~\eqref{eq:attribute-factoring} is a preprocessing step that helps us to avoid a blow-up in expression size when merging projections (Sec.~\ref{sec:motivation} \textbf{P1}).
Rules~\eqref{eq:duplicate-remove} and~\eqref{eq:duplicate-remove-set} can be used to remove unnecessary duplicate elimination operators (\textbf{P2}). 
Bottleneck \textbf{P3} is addressed by removing operators that block join reordering:  Rules~\eqref{eq:duplicate-remove}, \eqref{eq:duplicate-remove-set}, and~\eqref{eq:window-function} remove such operators. Even if such operators cannot be removed, Rule~\eqref{eq:pulling-up-provenance-projections} and Rule~\eqref{eq:remove-redundant-columns1} remove attributes that are not needed which reduces the schema size of intermediate results. 
\textbf{P4} can be addressed by using Rules~\eqref{eq:duplicate-remove}, \eqref{eq:duplicate-remove-set}, and \eqref{eq:window-function} to remove redundant operators. Furthermore, Rule~\eqref{eq:remove-redundant-columns1} 
removes columns that are not needed.
In addition to the rules discussed so far, we apply standard equivalences, because our transformations often benefit from these equivalences and they also allow us to further simplify a query. For instance, we apply \textit{selection move-around} (which benefits from the EQ property), merging of selections and projections (only if this does not result in a significant increase in expression size), and remove redundant projections (projections on all input attributes). 

\subsection{Correctness}
\label{sec:correctness}

\begin{Theorem}
The PATs from Fig.~\ref{fig:algebraic-rules} are equivalence preserving.  
\end{Theorem}
\begin{proof}
\underline{Rule~\eqref{eq:pulling-up-provenance-projections}:}
The value of attribute $b$ is the same as the value of $a$ (follows from $a \to b$). Since $b$ is not needed to evaluate $\Diamond(\projection_A(R))$, we can delay the computation $b$ after $\Diamond$ has been evaluated.
\underline{Rule~\ref{eq:duplicate-remove}:}
Since $keys(R) \neq \emptyset$, by Def.~\ref{def:def_keys} it follows that no duplicate tuples exist in $R$ ($t^n \in R \rightarrow n \leq 1$). Thus, we get $\duplicate (R) \rightarrow R$.
\underline{Rule~\ref{eq:duplicate-remove-set}:}
We say an operator $\Diamond(R)$ is insensitive to duplicates if for all $R$ we have $t \in \Diamond(\delta(R)) \leftrightarrow t \in \Diamond(R)$. That is, which tuples are returned by the operator is independent of input tuple multiplicities.
Since $set(\duplicate(R))$ = true, we know that there exists $op' = \delta$ with $op \ancestor op'$ and $\forall op'': \duplicate(R) \ancestor op'' \ancestor op'$ such that $op'' \not\in \{ \aggregation, \omega \}$. Note that operator types other than $\aggregation$ and $\omega$ are insensitive to duplicates. Thus, the set of tuples in the input of $op'$ is not affected by the rewrite $\duplicate(R) \to R$. While the multiplicities of these tuples can be affected by the rewrite, the final result of $op' = \duplicate$ is not affected. 
\underline{Rule~\ref{eq:remove-redundant-columns1}:}
Suppose $A=icols(R)$, by definition~\ref{def:def_icols} we get $R \rightarrow \projection_A (R)$.
\underline{Rule~\ref{eq:window-function}:}
From $x \not\in icols(\omega_{f(a) \to x} (R))$ follows $Q[\omega_{f(a) \to x} (R) \gets \projection_{\schema{R}}(\omega_{f(a) \to x} (R))] \equiv Q$. Based on the definition of $\omega$ it follows that $t^n \in \projection_{\schema{R}}(\omega_{f(a) \to x} (R)) \leftrightarrow t^n \in R$. Thus, $Q[\omega_{f(a) \to x} (R) \gets R] \equiv Q$.
\underline{Rule~\eqref{eq:attribute-factoring}:}
Let $e_1' = A + \eIf{\theta}{c}{0}$.
We distinguish two cases if $\theta$ holds, then both $e_1$ and $e_1'$ evaluate to $A+c$. If $\neg \theta$ holds, then both $e_1$ and $e_1'$ evaluate to $A$.
\end{proof}

\begin{Definition}
We use $d$ to denote the depth of the query $Q$, e.g., if $Q = R$ while $R$ is a relation, then $d(Q) = 1$; If $Q = op(Q_1)$, then $d(Q)  = d(Q_1) + 1$; If $Q = op(Q_1,Q_2)$, then $d(Q) = max(d(Q_1),d(Q_2)) + 1$. We also define the depth of operators, for example, if $Q = \projection_A(\selection_{\theta}(R))$ with $d(Q)=3$, then $d(\projection_A) = 1$, $d(\selection_{\theta}) = 2$ and $d(R) = 3$.

\end{Definition}

\begin{Theorem} \label{thm:thm_icols_subset}
Assume $Q=op(Q_1)$ and $Q_1 = op_1(Q_2)$, then $icols(op) \subseteq icols(op_1)$ expect for the operator which introduced new attributes. Following we use $New(op)$ to denote the new introduced attributes of $op$. 

$$New(op)= 
\begin{cases}
\emptyset         &  op = \selection_{\theta}, \join_{a=b}, \crossprod, \duplicate, \union, \intersection, \difference \\
\{ b_1,...,b_n \} & op = \projection_{a_1 \to b_1,..., a_n \to b_n} \\
\{ F(a) \}           & op = _{b_{1},...,b_{n}} \aggregation _{F(a)}(R) \\
\{ f(a) \}           & op = \omega_{f(a) \to x, G\|O}(R) \\
\end{cases}$$


\end{Theorem}
\begin{proof}
If we remove the attributes $E \subseteq icols(parent)$ from the $icols(child)$, the parent operator would be broke since the child operator does not contain the attributes needed in its parent operator. Thus the $icols$ of a child of an operator have to be a superset of $icols$ of its parent except for attributes being produced by its parent, for example, $\projection_{(A+B \leftarrow C)}R$, $C$ is in the $icols(\projection)$ but does not exist in $R$.
\end{proof}

\begin{Theorem}
The Top-down inference of property icols shows in Table~\ref{tab:top-down-icols} is correct for any query.
\end{Theorem}
\begin{proof}
We are proving the theorem by induction over the depth of the operator of the query. 

\underline{BASE CASE:} 
Let $op_1$ be the operator of depth 1, we firstly prove the theorem for the icols of $op_1$ of arbitrary query Q is correct no matter what the type of $op_1$ is. Assume $Q=op_1(Q_{sub})$, since $d(op_1) = 1$, then $icols(op_1) = \schema{Q}$. By definition~\ref{def:def_icols}, $Q[Q_{sub} \leftarrow \projection_{icols(op_1)}(Q_{sub})] = Q[Q_{sub} \leftarrow \projection_{\schema{Q}}(Q_{sub})] \equiv Q$. Next we prove $\schema{Q}$ is the minimal set of attributes $E \in \schema{Q}$ such that $Q \equiv Q[Q_{sub} \leftarrow \projection_{\schema{Q}}(Q_{sub})]$. WLOG, assume $A \subseteq \schema{Q}$, then $\projection_{\schema{Q}/A}Q \not \equiv Q$, thus $\schema{Q}$ is the minimal set. 
Now assume we have proved icols for the operator of depth up to n of arbitrary query Q , we denote the operator in depth n as $op_n$ and use $icols(op_n)$ to denote its icols. Now let us prove it for the operator of depth up to n+1, WLOG, let $op_{n+1}$ be an operator with $d(op_{n+1}) = n+1$ and use $icols(op_{n+1})$ to denote its icols.
Since $icols(op_{n+1})$ is based on operator $op_n$, now we discuss different types of $op_n$. 

\underline{CASE 1 ($\selection$):} If $op_n = \selection$, by definition~\ref{def:def_icols}, $icols(op_{(n+1)})$ should be the minimal set of attributes $E \subseteq \schema{Q_{sub}}$ such that $Q \equiv Q[Q_{sub} \leftarrow \projection_{icols(op_{(n+1)})}(Q_{sub})]$ while $Q_{sub} = op_{(n+1)}(Q_{sub}')$ be a subquery of $Q$. 
By Theorem~\ref{thm:thm_icols_subset}, $icols(op_n) \subseteq icols(op_{(n+1)})$. Since in the selection operator, the attributes in its condition $\theta$ are needed to evaluate the operator (we denote the set of attributes as $cols(\theta)$). If we remove attributes in $cols(\theta)$, the condition would be ill-defined. Thus $icols(op_n) \cup cols(\theta)$ is the minimal set, now we proved Rule 2.

\underline{CASE 2 ($\projection$):} If $op_n = \projection$, the proof steps are similar to the proof of CASE 1, $icols(op_{(n+1)})$ should satisify the definition~\ref{def:def_icols}. By Theorem~\ref{thm:thm_icols_subset}, $\projection$ introduced new attributes, thus $icols(op_n)$ should contain the attributes which used to construct the attributes in $icols(op_n)$. For example, $\projection_{(A+B \to C)}(R(A,B))$, the $icols(\projection) = \{C\}$ such that $icols(R)$ must contain attributes $A, B$, otherwise $\projection$ would be ill-defined. Thus if we have a query like $\projection_ {A_{1} \rightarrow B_{1},...,A_{n} \rightarrow B_{n}(R)}$ where $A_i$ and $B_i$ are expresions and $i \in \{1,...,n\}$, if $B_i \in icols$, then $icols(R)$ should contain all of the attributes exist in $A_i$. Now we proved the Rule 3.




\underline{CASE 3 ($\join$):} If $op_n =$ $\join_{a=b}$, $icols(op_{(n+1)})$ should contain the attributes in the join condition, here is $\{a,b\}$, otherwise the join operator $op_n$ would be ill-defined. By Theorem~\ref{thm:thm_icols_subset}, $icols(R \join S) \subseteq icols(R) \union icols(S)$, but $icols(op_{(n+1)})$ ($icols(R)$ or $icols(S)$) should contain attributes in its own icols, otherwise $op_{(n+1)}$ would be ill-defined. 
Thus $(icols(op_{n}) \cup \{a,b\})  \intersection icols(op_{(n+1)})$ is the the minimum set for satisfying the definition~\ref{def:def_icols}, now we proved the Rule 4.

\underline{CASE 4 ($\crossprod$):} If $op_n = \crossprod$, same with the proof of CASE 4, the difference is without the need of the attributes in the join condition, now we proved Rule 5.

\underline{CASE 5 ($\aggregation$):} If $op_n = \aggregation$, by Theorem~\ref{thm:thm_icols_subset}, $ _{b_{1},...,b_{n}} \aggregation _{F(a)}(R) $ introduced new attributes $\{F(a)\}$. To satisfy the definition~\ref{def:def_icols}, $icols(op_{(n+1)})$ equal to the group-by attributes and function attributes which used to construct the $\schema{op}$, otherwise the aggregation operator $op_n$ would be ill-defined. Now we proved the Rule 6.

\underline{CASE 6 ($\duplicate$):} If $op_n = \duplicate$, to satisfy the definition~\ref{def:def_icols}, $icols(op_{(n+1)})$ have to equal to the $\schema{op_n}$, otherwise tuple losing may be caused. For example, for the query $Q = \projection_A(\duplicate(R(A,B)))$ while relation R contains the tuples \{(1,2),(1,3)\} and $icols(\projection) = \{A\}$, then the result of Q is \{1,1\}. Assume $icols(\duplicate) \neq \schema{R}$ which only contains attribute A, by definition~\ref{def:def_icols}, $Q= \projection_A(\duplicate(\projection_A(R(A,B))))$ whose result is \{1\}, one tuple was lost. 

\underline{CASE 7 ($\union $):} If $op_n = \union$, by Theorem~\ref{thm:thm_icols_subset} no new attribiutes was introduced and in this paper we only consider bag-semantics, thus $icols(op_n) = icols(op_{(n+1)})$, now we proved the Rule 8.

\underline{CASE 8 ($\difference$):} If $op_n = \difference$, to satisfy the definition~\ref{def:def_icols}, the $icols$ of its child operator have to equal to the schema of its child, otherwise tuple losing may be caused. For example, for the query $Q = \projection_A(\difference(R(A,B),S(C,D)))$, relation R contains data \{(1,2),(1,3)\}, relation S contains data \{(2,4)\}, the result of Q is \{1,1\}. $icols(\projection) = {A}$, if $icols(R(A,B)) \neq \schema{R}$ which only contains attribute A, by applying definition~\ref{def:def_icols}, query Q would be rewrited to $Q = \projection_A(\difference(\projection_A(R(A,B)),\projection_C(S(C,D))))$ whose result is \{1\}, another tuple \{1\} was lost. Thus we proved the Rule 10.

\underline{CASE 9 ($\intersection$):} If $op_n = \intersection$, the proof of Rule 10 is similar with CASE 8.

\underline{CASE 10 ($\omega$):} If $op_n = \omega$, the proof is similar with CASE 5 ($\aggregation$). Since without the attirubtes in G, O and the function, the operator $op_n$ would be ill-defined. 


\end{proof}



\begin{Theorem} \label{thm:thm-trans}
Except for the operator $\union$ and $\difference$, if two attributes are in the same $EC$ in the child operator and these attributes exist in the parent operator, then they are also should be in the same $EC$ in the parent operator. 
\end{Theorem}

\begin{proof}
Assume attributes $a$ and $b$ in the result schema, for each operator (except for the $\union$ and $\difference$), if the operator preserve that it is not manipulating its values (the values are the same in the input and the output), then if they are equals before, they are equals afterwards.
For $\selection, \join, \crossprod, \duplicate$ and $\intersection$, attributes are in the output of the operator are not modified the iterms of values, thus if they are equals before, they are equals afterwards. 
For $\projection$, if the attributes in the output igonring renaming, then this holds.  
For $\aggregation$ and $\omega$, the attributes in the result are only the part of group by attributes which are not manipulated, thus this holds.
For $\union$, this does not hold. For example, relation $R$ has attributes a and b and relation $S$ has attributes c and d, $EC(R) = \{\{a,b\}\}$ and $EC(S) = \{\{c\},\{d\}\}$. If $R$ contains the value $\{(1,1)\}$ and $S$ contains the value $\{(3,4)\}$, then $R \union S = \{(1,1),(3,4)\}$. It is obviously that $EC(R) = \{\{a,b\}\}$ does not hold in $R \union S$. The same for the operator $\difference$. 

\end{proof}

\begin{Theorem}
Bottom-up inference of property EC in Table~\ref{tab:bottom-up} is correct.
\end{Theorem}
Recall that we are not claiming that our EC are maximum, becasue when the definition~\ref{def:def_ec} holds, definitely the attributes are being in the same subset of EC. For example, for a query $Q=\projection_{f(a) \to c, g(b) \to d}(R)$, even through $c=d$, but the equivalence of $f(a)$ and $g(b)$ is undecidable. Thus in the following proof, we only prove our rules are correct.

\begin{proof}
\underline{BASE CASE:} Assume a query Q, $d(Q)=n$ and $d(op_n) = n$, we firstly prove the theorem for the EC of $op_n$ is correct of arbitrary query Q is correct. Since $op_n$ must be a relation operator, each element in $EC(op_n)$ only contains one set which only has one attribute, then we get the Rule 1. Now assume we have proved EC for the operator of depth up to n-k from n of arbitrary query Q, we denote the operator in depth n-k as $op_{n-k}$. Now let us prove it for the operator of depth up to n-(k+1) from n, WLOG, let $op_{n-(k+1)}$ be an operator with $d(op_{n-(k+1)}) = n - (k+1)$. Here we consider different type of operator for $op_{n-(k+1)}$.

\underline{CASE 1 ($\selection$):} if $op_{n-(k+1)} = \selection$, 
for every $a=b$ exist as conjunction in $\theta_i$, a and b should be in the same subset of EC. Assume there exists the result tuple of the $\selection$ where $a \neq b$, then obviously this does not fulfill this $\theta_i$, since this is conjunctive condition, the result tuple does not fulfill the selection condition which is contradiction.

\underline{CASE 2 ($\projection$):} if $op_{n-(k+1)} = \projection$, since it is undecidable for the equivalence of expressions, it is unsafe to put new attributes derived by expressions into the same subset of EC. 
By theorem~\ref{thm:thm-trans} we get Rule 3.

\underline{CASE 3 ($\join$):} if $op_{n-(k+1)} =$ $\join$, it is same with \underline{CASE 1 ($\selection$)}, we get Rule 4.

\underline{CASE 4 ($\crossprod$):} if $op_{n-(k+1)}  = \crossprod$, since no new equality was introduced, by theorem~\ref{thm:thm-trans} we get Rule 5.

\underline{CASE 5 ($\aggregation$):} if $op_{n-(k+1)} = \aggregation$, since $\aggregation$ does not change the value of group by attributes and the value of the attributes derived by function can not be predetermined. By theorem~\ref{thm:thm-trans} we get the Rule 6.

\underline{CASE 6 ($\duplicate$):} if $op_{n-(k+1)} = \duplicate$, since $\duplicate$ does not change the value of its attributes and by theorem~\ref{thm:thm-trans} we get Rule 7.

\underline{CASE 7 ($\union$):} if $op_{n-(k+1)} = \union$, since the result tuples of $\union$ either from its left child or from its right child, its EC can only choose the common part by ignoring the renaming, by theorem~\ref{thm:thm-trans} we get Rule 8.

\underline{CASE 8 ($\intersection$):} if $op_{n-(k+1)} = \intersection$, since the result tuples of $\intersection$ exists both in its left child and right child, ignoring the renaming its EC satisifys its both children. Thus by theorem~\ref{thm:thm-trans} we get Rule 9.

\underline{CASE 9 ($\difference$):} if $op_{n-(k+1)} = \difference$, since the result tuples of $\difference$ only from its left child. Thus by theorem~\ref{thm:thm-trans} we get Rule 10.

\underline{CASE 10 ($\omega$):} if $op_{n-(k+1)} = \omega$, since new attribute was intorduced from the function and its value can not be predetermined. Then by theorem~\ref{thm:thm-trans} we get Rule 10.

\end{proof}

For the \underline{CASE 2 ($\projection$)},  it is unsafe to add more into the EC since it is undeciable for the expression. For example, for a query $Q=\projection_{f(a) \to c, g(b) \to d}(R)$, even through $c=d$, but the equivalence of $f(a)$ and $g(b)$ is undecidable.
For the \underline{CASE 7 ($\union$)}, it is unsafe to add more into the EC since the result tuples of $\union$ either from its left child or from its right child. For example, for q query $Q = R \union S$ where relation R(a,b) contains the tuples $\{(1,1)\}$ and S(c,d) contains the tuples $\{(1,2)\}$. Then we know $a=b$ in R, but add \{a,b\} into $EC(R \union S)$ is unsafe, since $c \neq d$ in relation S. 


\begin{Theorem} \label{thm: thm-sel-push}
If a selection can be pushed down through the parent operator to the child operator where this selection is build based on the attributes in the EC(parent), the attrbutes exist in the same set of EC(parent) also should be in a same set of EC(child). 
\end{Theorem}
For example, a query $Q=\projection_{a,b}(\projection_{a,b,c}(R))$ with $EC(\projection_{a,b}) = \{\{a,b\}\}$. By definition~\ref{def:def_ec} we can add a selection $\selection_{a=b}$ on query Q such that $Q=\selection_{a=b}(\projection_{a,b}(\projection_{a,b,c}(R)))$. Since pushing down $\selection_{a,b}$ does not affect the result of query Q, then $Q=\projection_{a,b}(\selection_{a=b}(\projection_{a,b,c}(R)))$. By definition~\ref{def:def_ec}, $\{a,b\} \in E$ where $E \in EC(\projection_{a,b,c})$.

\begin{proof}
Reconsider above example, assume after pushing down $\selection_{a=b}$, for any  $E \in EC(\projection_{a,b,c})$, $\{a,b\} \not \in E$ which is contradict with definition~\ref{def:def_ec}, thus proof done. 
\end{proof}

\begin{Theorem} 
Top-down inference of property EC in Table~\ref{tab:top-down} is correct.
\end{Theorem}
\begin{proof}
Assume operator $op$ is the parent of operator $op_{child}$ and $EC(child)$ is based on $EC(op)$, 
thus now we consider different type of operator for $op$. (If $op$ contains two children, we split $EC(op_{child})$ as $EC(op_{left})$ and $EC(op_{right})$). 

\underline{CASE 1 ($\selection$):} if $op = \selection$, since any selection can be pushed down throught a selection, by theorem~\ref{thm: thm-sel-push} proof done, we get the Rule 1.


\underline{CASE 2 ($\projection$):} if $op = \projection$, by reverse the renaming, selection can be pushed down. By theorem~\ref{thm: thm-sel-push} proof done, we get the Rule 2.

\underline{CASE 3 ($\join$):} if $op =$ $\join$, if both attributes in a selection exist in the same child, then this selection can be pushed down. By theorem~\ref{thm: thm-sel-push} proof done, we get the Rule 3.


\underline{CASE 4 ($\crossprod$):} if $op = \crossprod$, the same with CASE 3, then we get Rule 4.

\underline{CASE 5 ($\aggregation$):} if $op = \aggregation$, the selection can be pushed down where its attributes exist in the group by, then by theorem~\ref{thm: thm-sel-push} proof done, we get the Rule 5.


\underline{CASE 6 ($\duplicate$):} if $op = \duplicate$, any selection can be pushed down throught a $\duplicate$, by theorem~\ref{thm: thm-sel-push} proof done, we get the Rule 6.


\underline{CASE 7 ($\union$):} if $op = \union$, since the Rule 8 in Table~\ref{tab:bottom-up} the selection can be pushed down to both children ($op_{left}$ and $op_{right}$) of $op$. By theorem~\ref{thm: thm-sel-push} proof done, we get the Rule 7.


\underline{CASE 8 ($\intersection$):} if $op = \intersection$, since the result tuples of $op$ are exist in both $op_{left}$ and $op_{right}$, the selection can be pushed to to the $op_{left}$ and after reversing the renaming the selection can be pushed to to the $op_{right}$. By theorem~\ref{thm: thm-sel-push} proof done, we get the Rule 8.


\underline{CASE 9 ($\difference$):}, if $op = \difference$, the selection can only be pushed down to the left child of $op$ ($op_{left}$) since the tuples in $op_{left} \difference op_{right}$ does not exist in $op_{right}$. By theorem~\ref{thm: thm-sel-push} proof done, we get the Rule 9.

\underline{CASE 10 ($\omega$):}, if $op = \omega$, the selection can be pushed down where its attributes exist in the O and G	, then by theorem~\ref{thm: thm-sel-push} proof done, we get the Rule 10.

\end{proof}

It is worth to mention that at last we union the result of each case with EC(child) to avoid losing any equivlent attributes and applying the ${\cal E}^* $ to remove the duplicated sets.

\begin{Theorem}
Bottom-up inference of property set in Table~\ref{tab:top-down-set} is correct.
\end{Theorem}
\begin{proof}
\underline{BASE CASE:} Let $op_1$ be the operator of depth 1, we firstly prove the theorem for the set of $op_1$ of arbitrary query Q is correct no matter what the type of $op_1$ is. By definition\ref{def:def_set}, since $set(op) = true$, we initialize $set(op_1) = false$. Now assume we have proved set for the operator of depth up to n of arbitrary query Q, we denote the operator in depth n as $op_n$ and use $set(op_n)$ to denote its set. Now let us prove it for the operator of depth up to n+1, WLOG, let $op_{n+1}$ be an operator with $d(op_{n+1}) = n + 1$ and use $set(op_{n+1})$ to denote its set. Since $op_{n+1}$ is based on $op_n$, we discuss different types of $op_n$. 

\underline{CASE 1 ($\{\selection, \projection, \join, \crossprod, \union, \intersection, \difference\}$):} If $op_n \in \{\selection, \projection, \join, \crossprod, \union, \intersection, \difference\}$, by definition~\ref{def:def_set} only $\{\aggregation, \omega\}$ affect the set property, $set(op_{n+1})$ have to keep the same with the $set(op_n)$. Thus we initialize $set(op) = true$ while $d(op) \neq 1$, then $set(op_n) \wedge set(op_{n+1})$ can be used to keep the $new$ $set(op_n) = set(op_{n+1})$. 

\underline{CASE 2 ($\duplicate$):} If $op_n = \duplicate$, by definition~\ref{def:def_set}, the set property of $\duplicate$ in the downstream of $op_n$ have to be true if no $\{\aggregation, \omega\}$ exist by using the $new$ $set(op_{n+1}) = set(op_{n+1}) \wedge true$. 

\underline{CASE 3 ($\{\aggregation, \omega\}$):} If $op_n \in \{\aggregation, \omega\}$, by definition~\ref{def:def_set}, we have to keep the set property of $\duplicate$ in the downstream of $op_n$ equal to false by using  $set(op_{n+1}) = false$.

\end{proof}

\begin{Theorem}
Bottom-up inference of property key in Table~\ref{tab:bottom-up-key} is correct.
\end{Theorem}

Recall that we are not claiming that our EC are maximum, such that we can not guarantee that our key property is complete and if it is a candidate key. For example, for a query $Q=\projection_{f(a) \to c, g(b) \to d}(R)$, even through $c=d$, but the equivalence of $f(a)$ and $g(b)$ is undecidable. 

\begin{proof}
Assume operator $op$ is the parent of operator $op_{child}$ and $key(op)$ is based on $key(op_{child})$, 
we consider different type of operator for $op$. (If $op$ contains two children, we split $key(op_{child})$ as $key(op_{left})$ and $key(op_{right})$)

\underline{CASE 1 ($\selection$):} if $op = \selection$, since the output tuples of $\selection$ is a subset of its input, if every tuples in the input is satisify the definition~\ref{def:def_keys}, then the output is also satisifiy the definition~\ref{def:def_keys}. Thus we get the Rule 1.

\underline{CASE 2 ($\projection, \join, \crossprod$):} if $op = \projection$ or $\join$ or $\crossprod$, these three operators are already proved in Section 5.2.1 of Paper~\cite{SD96}.



\underline{CASE 3 ($\aggregation$):} if $op = \aggregation$, since the new key can be added via GROUP BY operation and the definition\ref{def:def_keys} holds. However, if the original key is a subset of the new key, we should keep the original one. We get the Rule 5


\underline{CASE 4 ($\duplicate$):} if $op = \duplicate$, since $\duplicate$ only removes the duplicated tuples such that the output tuples of $\duplicate$ is a subset of its input, if every tuples in the input is satisify the definition~\ref{def:def_keys}, then the output is also satisifiy the definition~\ref{def:def_keys}. But if no key exist in the input, 
every output attributes should be the key. Thus we get the Rule 6.


\underline{CASE 5 ($\union$):} if $op = \union$, since after $\union$, the output may contain the duplicated tuples such that the definition~\ref{def:def_keys} does not hold. Thus no key exist, we get Rule 7.


\underline{CASE 6 ($\intersection$):} if $op = \intersection$, since the output tuples of $\intersection$ is a subset of its left input tuples and also be a subset of its right input tuples, whether the $key(op) = key(op_{left})$ or $key(op) = key(op_{right})$, the definition~\ref{def:def_keys} holds, we get Rule 8.

\underline{CASE 7 ($\difference$):}, since $\difference$ only removes some tuples of its left child such that the output tuples of $\difference$ is a subset of its left input, if every tuples in the left input is satisify the definition~\ref{def:def_keys}, then the output is also satisifiy the definition~\ref{def:def_keys}, thus when $key(op) = key(op_{left})$,  definition~\ref{def:def_keys} holds, we get Rule 9.

\underline{CASE 8 ($\omega$):}, the output tuples of $\omega$ is a subset of the input tuples only with one additional attribute derived from the function, such that if every tuples in the input is satisify the definition~\ref{def:def_keys}, then the output is also satisifiy the definition~\ref{def:def_keys}. Thus we get the Rule 10.
\end{proof}





\section{Instrumentation Choices}
\label{sec:transf-appl-during}

\parttitle{Window vs. Join}\label{sec:window-vs-join}
The \textit{Join} method for instrumenting an aggregation operator for provenance capture was first used by Perm~\cite{glavic2013using}. To propagate provenance from the input of the aggregation to produce results annotated with provenance, the original aggregation is computed and then joined with the provenance of the aggregations input on the group-by attributes. This will match the aggregation result for a group with the provenance of tuples in the input of the aggregation that belong to that group (see~\cite{glavic2013using} for details). For instance, $_{b}\aggregation_{sum(a)}(R)$ with ${\bf R} = (a,b)$ would be rewritten into $\projection_{b,sum(a),P(a),P(b)}(_{b}\aggregation_{sum(a)}(R) \join_{b = b'} \projection_{b\to b', a \to P(a), b \to P(b)}(R))$. Alternatively, the aggregation can be computed over the input with provenance using the window operator $\omega$ by turning  
the group-by into a partition-by. The rewritten expression is $\projection_{b,sum(a),P(a),P(b)}(_{b \|}\omega_{sum(a)}(R)(\projection_{b\to b', a \to P(a), b \to P(b)}(R)))$.
The \textit{Window} method has the advantage that 
no additional joins are introduced. However, as we will show in Sec.~\ref{sec:experiments}, the \textit{Join} method is superior in some cases and, thus, the choice between these alternatives should be cost-based.


\parttitle{FilterUpdated vs. HistJoin}
Our approach for capturing the provenance of a transaction $T$~\cite{AG16a,AG17} only returns the provenance of tuples that were affected by $T$. We consider two alternatives for achieving this. 
The first method is called \textit{FilterUpdated}.  Consider a transaction $T$ with $n$ updates and let $\theta_i$ denote the condition (\lstinline!WHERE!-clause) of the $i^{th}$ update. Every tuple updated by the transaction has to fulfill at least one $\theta_i$. Thus, this set of tuples can be computed by applying a selection on condition $\theta_1 \vee \ldots \vee \theta_n$ to the input of reenactment. The alternative called \textit{HistJoin} uses time travel to determine based on the database version at transaction commit which tuples where updated by the transaction. It then joins this set of tuples with the version at transaction start to recover the original inputs of the transaction. For a detailed description see~\cite{AG16a}.
\textit{FilterUpdated} is typically superior, because it avoids the join applied by \textit{HistJoin}. However, for transactions with a large number of operations or complex \lstinline!WHERE!-clause conditions, the cost of evaluating the selection condition $\theta_1 \vee \ldots \vee \theta_n$ can be higher than the cost of the join.


\begin{algorithm}[t]
  \caption{CBO}
  \label{alg:cbo-skeleton}
  \begin{algorithmic}[1]
    \Procedure{CBO}{$Q$}
      \State $T_{best} \gets \infty$, $T_{opt} \gets 0.0$
      \While {$\Call{hasMorePlans}{ } \wedge \Call{continue}$}
        \State $t_{before} \gets \Call{currentTime}{ }$
        \State $P \gets \Call{generatePlan}{Q}$ 
        \State $T \gets \Call{getCost}{P}$
        \If {$T<T_{best}$}
          \State $T_{best} \gets T$
          \State $P_{best} \gets P$
        \EndIf
        \State \Call{genNextIterChoices}{ }
        \State $t_{after} \gets \Call{currentTime}{ }$
        \State $T_{opt} = T_{opt} + (t_{after} - t_{before})$
      \EndWhile
      \State \Return $P_{best}$
    \EndProcedure
        
  \end{algorithmic}
\end{algorithm}

\section{Cost-based Optimization}\label{sec:cbo}

\parttitle{CBO Algorithm} The pseudocode for our CBO algorithm is shown in Algorithm~\ref{alg:cbo-skeleton}. The algorithm consists of a main loop that is executed until the whole plan space has been explored (function $\Call{hasMorePlans}{}$) or until a stopping criterion has been reached (function $\Call{Continue}{}$). 
In each iteration, function $\Call{generatePlan}{}$ takes the output of the parser and runs it through the instrumentation pipeline (e.g, the one shown in Fig.~\ref{fig:cbo-arch}) to produce an SQL query. The pipeline components inform the optimizer about choice points using function $\Call{makeChoice}{}$. The resulting plan $P$ is then costed. 
If the cost $T$ of plan $P$ generated in the current iteration is less than the cost $T_{best}$ of the best plan found so far, then $P$ becomes the next best plan. Finally, we decide which optimization choices to make in the next iteration using function $\Call{genNextIterChoices}{}$.
While on the surface our CBO algorithm resembles standard CBO, the difference lies in implementation of the $\Call{generatePlan}{}$ and $\Call{genNextIterChoices}{}$ functions and their interaction with the optimizer's callback API described below. This API achieves decoupling of enumeration, costing, and plan space traversal order from plan construction. As will be explained further in Sec.~\ref{sec:search-space} this enables the optimizer to enumerate all plans for a blackbox  instrumentation pipeline. 
New choices are discovered at runtime when a step in the pipeline informs the optimizer about a choice point. 

\parttitle{Costing} Our default cost estimation implementation uses the DBMS to create an optimal execution plan for $P$ and estimate its cost. 
This ensures that we get the estimated cost for the plan that would be executed by the backend instead of estimating cost based on the properties of the query alone. However, this is not a requirement in our framework, e.g., a separate costing module may be used for backends that do not apply CBO. 

\parttitle{Search Strategies}
Different strategies for exploring the plan space 
are implemented as different versions  of the $\Call{continue}{}$, $\Call{genNextIterChoices}{}$, and $\Call{makeChoice}{}$ functions. 
The default 
setting guarantees that the whole search space will be explored ($\Call{continue}{}$ returns true). 
Our CBO algorithm keeps track of how much time has been spent on optimization so far ($T_{opt}$) which 
may be used to decide when to stop optimization.
\subsection{Registering Optimization Choices}



We want to make the optimizer aware of choices available in an instrumentation pipeline without having to significantly change existing code.
%
This is achieved by registering choices though a callback interface. 
Thus, it is easy to introduce new choices in any step of an instrumentation pipeline at runtime by adding calls to the optimizer's $\Call{makeChoice}{}$ function
without any modifications to the optimizer and only trivial changes to the code containing the choice.
The callback interface has two purposes: 1) telling the optimizer about when a choice needs to be made and how many alternatives to choose from and 2) allowing it to control which options are chosen.   
%
Recall that we refer to a point in the code where a choice is enforced as 
a \textit{choice point}. 
A choice point 
has a fixed 
number of \textit{options}, 
the optimizer's callback function returns an integer indicating to the caller which option should be taken. 
%

\begin{Example}
  Assume a provenance engine implements the \emph{Join} and \emph{Window} methods as functions \lstinline!instAggJoin! and \lstinline!instAggWin!.
To make a cost-based choice between these methods, we call $\Call{makeChoice}{}$. The parameter passed to this function is the number $n$ of 
choices ($n=2$ in this example).  
$\,$\\[-5mm]
\begin{lstlisting}[language=c] 
if (makeChoice(2) == 0) instAggWin(Q);
else instAggJoin(Q);
\end{lstlisting}
$\,$\\[-8mm]
The optimizer responds 
by returning a number between 0 and $n-1$ representing the choice to be taking. In our example, we would use the  \emph{Window} method if 0 is returned. 
\end{Example}

During one iteration a code fragment containing a call to $\Call{makeChoice}{}$ may be executed several times.
The optimizer treats every call as an independent choice point, e.g., 4 possible combinations of the \emph{Join} and \emph{Window} methods will be considered for instrumenting a query if it has two aggregations.

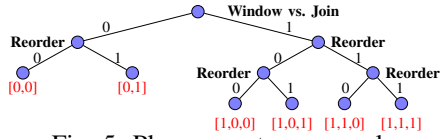
\begin{figure}[t]
  \centering
  $\,$\\[-7mm]
\resizebox{!}{0.3\columnwidth}{
\begin{tikzpicture}
[every node/.style={circle,draw,fill=blue!50,label position={east},inner sep={1mm}},
el/.style={draw=none,fill=none,inner sep={1mm}},
level distance=7mm,
level 1/.style={sibling distance=55mm},
level 2/.style={sibling distance=25mm},
level 3/.style={sibling distance=13mm}
]
  \node[label={[label distance=0.4cm]{\textbf{Window vs. Join}}}] {} 
        child {node[label=left:{\textbf{Reorder}}] {} 
            child {node[label={[label distance=-0.3cm]below:{\textcolor{red}{[0,0]}}}] {} 
              edge from parent node[left,el] {0}
            }
            child {node[label={[label distance=-0.3cm]below:{\textcolor{red}{[0,1]}}}] {} 
              edge from parent node[right,el] {1}
            }
            edge from parent node[left,el,inner sep={4mm}] {0}
        }
        child {node[label={\textbf{Reorder}}] {}
          child {node[label=left:{\textbf{Reorder}}] {} 
            child {node[label={[label distance=-0.3cm]below:{\textcolor{red}{[1,0,0]}}}] {} 
              edge from parent node[left,el] {0}
            }
            child {node[label={[label distance=-0.3cm]below:{\textcolor{red}{[1,0,1]}}}] {} 
              edge from parent node[right,el] {1}
            }
            edge from parent node[left,el] {0}
          }
          child {node[label=right:{\textbf{Reorder}}] {} 
            child {node[label={[label distance=-0.3cm]below:{\textcolor{red}{[1,1,0]}}}] {} 
              edge from parent node[left,el] {0}
            }
            child {node[label={[label distance=-0.3cm]below:{\textcolor{red}{[1,1,1]}}}] {} 
              edge from parent node[right,el] {1}
            }
            edge from parent node[right,el] {1}            
          }
          edge from parent node[right,el,inner sep={4mm}] {1}
        }
  ;
\end{tikzpicture}
}\\[-5mm]
\caption{Plan space tree example}
\label{fig:plan-tree-example}
\end{figure}

\subsection{Plan Space}
\label{sec:search-space}

We now 
look at the shape of the search space for given a query and set of choice points. 
During one iteration we may hit any number of choice points and each choice made 
may effect what other choices have to be made in the remainder of this iteration. 
We use a data structure called \textit{plan tree} that models the plan space shape. In the plan tree each intermediate node represents a choice point, outgoing edges from a node are labelled with options and children represent choice points that are hit next. 
A path from the root of the tree to a leaf node represents a particular sequence of choices that results in the plan represented by this leaf node. 

\begin{Example}
  Assume we use the two choice points: 1) Window vs. Join; 2) reordering join inputs. The second choice point can only be hit if a join operator exist, e.g., if we choose to use the \emph{Window} method then the resulting algebra expression may not have any joins and this choice point would never be hit.
 Assume we have to instrument a query which is an aggregation over the result of a join.   
Fig.~\ref{fig:plan-tree-example} shows the corresponding plan tree. When instrumenting the aggregation, we have to decide whether to use the \textit{Window} (0) or the \textit{Join}
method (1). If we choose 
(0), then we can still decide wether to reorder the inputs of the join or not. If we 
choose (1), then there is an additional join for which we have to decide whether to reorder its input. 
%
The tree is asymmetric, i.e., the number of choices to be made in each iteration (path in the tree) is not constant. 
\end{Example}

\begin{algorithm}[t]
  \caption{Default \textsc{makeChoice} Function}
  \label{alg:callback}
  \begin{algorithmic}[1]
    
    \Procedure{makeChoice}{$numChoices$}
      \If {$len(p_{next}) > 0$}
        \State $choice \gets popHead(p_{next})$
      \Else
        \State $choice \gets 0 $
      \EndIf
        \State $p_{cur} \gets p_{cur} \cList choice$
        \State $n_{opts} \gets n_{opts} \cList numChoices$
      \State \textbf{return} choice
    \EndProcedure
        
  \end{algorithmic}
\end{algorithm}

\begin{algorithm}[t]
  \caption{Default \textsc{genNextIterChoices} Function}
  \label{alg:gen-next-option}
  \begin{algorithmic}[1]
    
    \Procedure{genNextIterChoices}{ }
      \State $p_{next} \gets p_{cur}$
      \For {$i \in \{ len(p_{next}), \ldots, 1\}$} 
        \State $c \gets popTail(p_{next})$
        \State $nops \gets popTail(n_{opts})$
        \If{$c+1 < nops$}
          \State $c \gets c+1$
          \State $p_{next} \gets p_{next} \cList c$
          \State \textbf{break}
        \EndIf
      \EndFor    
      \State $p_{cur} \gets []$
      \State $n_{opts} \gets []$
    \EndProcedure
    
  \end{algorithmic}
\end{algorithm}

\subsection{Plan Enumeration}\label{sec:plan-enumerate}

While the plan space tree encodes all possible plans for a given query and set of choice points, it would not be feasible to fully materialize it, because its size can be exponential in the maximum number of choice points that are  hit during one iteration (the depth $d$ of the plan tree).  
Our default implementation of the $\Call{generateNextPlan}{}$ and $\Call{makeChoice}{}$ functions 
explores the whole plan space using ${\cal O}(d)$ space. 
As long as we know the path taken in the previous iteration (represented as a list of choices as shown in Fig.~\ref{fig:plan-tree-example}) and for each node (choice point) on this path the number of available options, then we can determine what choices should be made in the next iteration to reach the leaf node (plan) immediately to the right of the previous iteration's plan.  
If $p_{cur}$ is the path explored in the previous iteration, then by taking the next available choice as late as possible on the path will lead to the next node on the leaf level.  
Let $p_{next}$ be the prefix of $p_{cur}$ that ends in the new choice to be taken at the latest possible step. If following $p_{next}$  leads to a path is longer than $p_{next}$, then after making $len(p_{next})$ choices the first option should be chosen for the remaining choice points. 

\begin{Example}
Reconsider the plan tree shown in Fig.~\ref{fig:plan-tree-example} and assume the plan created in the previous iteration is $[0,1]$. The second choice point hit does have no additional options, but the first 
one does have one additional option. That is, to reach the next leaf node to the right of $p_{cur} = [0,1]$, we should choose the other option $p_{next} = [1]$. This option leads us to a path that is longer than $len(p_{next}) = 1$, thus we choose the first option for all remaining choice points leading us to the leaf node $[1,0,0]$. In the next iteration we have one more choice available in the last step of the path leading us to $[1,0,1]$. This process continues until no more choices are available.
\end{Example}
%
%
We use square brackets to denote lists, e.g., $p_{next}=[0,1]$ denotes a list with two element $0$ and $1$. We use $[]$ to represent an empty list.
$L \gets L \cList e$ denotes appending element $e$ to list $L$. Functions $popHead(L)$ and $popTail(L)$ remove and return the first (respective last) element of list $L$.

\parttitle{The makeChoice Function} Algorithm~\ref{alg:callback} shows 
as long as there are predetermined choices available (list $p_{next}$) we pick the next choice from this list. If list $p_{next}$ is empty, then we pick the first choice ($0$). In both cases the choice is appended to the current path and the number of available options for the current choice point is appended to list $n_{opts}$.

\parttitle{Determining Choices for the Next Iteration} 
Algorithm~\ref{alg:gen-next-option} determines which options to be picked in the next iteration. We copy the path from the previous iteration (line 2) and then repetitively remove elements from the tail of the path and from the list storing the number of options ($n_{ops}$) until we have removed an element $c$ for which at least one more alternative exits ($c + 1 < n_{ops}$). Once we have found such an element we append $c+1$ as the new last element to the path.

Given a set of choice points, our algorithm is guaranteed to enumerate all plans for an input. 

\begin{Theorem}
Let $Q$ be input query. Algorithm~\ref{alg:cbo-skeleton} iterates over all plans that can be created for the given choice points.
\end{Theorem}
\subsection{Alternative Search Strategies} \label{sec:traversal-strategies}

\parttitle{Traversal Order}
A simple solution for dealing with large search spaces is to define a threshold $\stopThresh$ for the number of iterations and stop optimization once this threshold is reached. However, the search space traversal strategy we have introduced in Sec.~\ref{sec:plan-enumerate} (we call it \textit{sequential-leaf-traversal} in the following) is not suited well for this solution. 
Since it only changes one choice at a time, the plans explored by this strategy for a  plan space that is large compared to threshold $\stopThresh$ are likely quite similar. To address this problem, we have developed a second  strategy which we call \textit{binary-search-traversal}. This strategy approximates a binary search over the leaves of a plan tree. The method maintains a queue of intervals (pairs of paths in the plan tree) initialized with the path to the left-most and right-most leaf of the plan tree. The strategy repeats the following steps until all plans have been explored: 1) fetch an interval $[P_{low},P_{high}]$ from the queue; 2) compute a path that is approximately a prefix of the path to the leaf that lies in the middle of the interval; 3) extend this path to create a plan $P_{middle}$; and 4) push two new intervals to the end of the queue: $[P_{low}, P_{middle}]$ and $(P_{middle}, P_{high}]$. 

\parttitle{Simulated Annealing}
Metaheuristics such as simulated annealing and genetic algorithms have a long tradition in query optimization to deal with large search spaces, e.g., some systems apply metaheuristics for join enumeration once the number of joins exceeds a threshold 
or for cost-based query transformations~\cite{AL06}. 
We have implemented the \textit{Simulated Annealing} metaheuristic. 
This method starts from a randomly generated plan and 
in each step applies a random transformation  to derive a plan $P_{cur}$ from the previous plan $P_{pre}$ (let $C_{cur}$ and $C_{pre}$ denote the costs of these plans). 
If $C_{cur} <C_{pre}$, $P_{cur}$ is used as $P_{pre}$ for the next iteration. 
Otherwise, the choice of whether to discard $P_{cur}$ or use it as the new $P_{pre}$ is made probabilistically. The probability depends on the cost difference $C_{cur} - C_{pre}$ and a parameter $temp$ called the temperature which is decreased over time based on a cooling  rate. Initially, the probability to choose an inferior plan is higher  to avoid getting stuck in a local minima early on. By decreasing the temperature (and, thus also probability) over time, the approach will converge eventually. 
\parttitle{Balancing Optimization vs. Runtime}\label{sec:balance-opt}
%
All strategies discussed so far have the disadvantage that they do not adapt the effort spend on optimization based on how expensive the input query is. Obviously, spending more time on optimization than on execution is undesirable (assuming that provenance requests are typically ad hoc). Ideally, we would like to minimize the sum of the time spend on optimization ($T_{opt}$) and the execution time of the best plan $T_{best}$ by stopping optimization once a cheap enough plan has been found. This is obviously an online problem
, i.e., after each iteration we can decide to either execute the current best plan or continue to produce more plans with the hope to discover a better plan in future iterations. 
The following simple stopping condition results in a 2-competitive algorithm (i.e., $T_{opt} + T_{best}$ is guaranteed to be less than 2 times the minimal achievable cost)  if the time spend in each iteration is bound by a  constant: stop iteration if $T_{best} < T_{opt}$.

\section{Related Work}\label{sec:relatedwork}

Our work is related to optimization techniques that sit on top of standard CBO, to other approaches for compiling non-relational languages into SQL, and to optimization of provenance capture and storage.

\parttitle{Cost-based Query Transformation}
State-of-the-art DBMS apply transformations such as decorrelation of nested subqueries~\cite{SP96} 
in addition to (typically exhaustive) join enumeration and choice of physical operators. Often such transformations are integrated with CBO~\cite{AL06} by iteratively rewriting the input query through transformation rules and then finding the best plan for each rewritten query. 
Typically, metaheuristics (randomized search) are applied to deal with the large search space. Extensibility of query optimizers has been studied in, e.g., ~\cite{graefe1993volcano}. 
While our CBO framework 
is also applied on-top of standard database optimization, we 
 can  turn any choice (e.g., ICs) within an instrumentation pipeline into a cost-based decision while cost-based query transformation is typically limited to algebraic transformations. 
Furthermore, our framework 
has the advantage  
that new optimization choices can be added without modifying the optimizer and with minimal changes to existing code.

\parttitle{Compilation of Non-relational Languages into SQL}\label{sec:rel-non-rel-to-rel}
Approaches that compile non-relational languages (e.g., XQuery~\cite{grust2010let,LK05}) or extensions of relational languages (e.g., temporal~\cite{SJ01} and nested collection models~\cite{CL14}) into SQL face similar challenges as we do.  
Grust et al.~\cite{grust2010let} optimize compilation of XQuery into SQL. The approach heuristically applies algebraic transformations  to cluster join operations 
with the goal to produce an SQL query that can successfully be optimized by a relational database. 
We adopt the idea of inferring properties over algebra graphs introduced in this work.  However, to the best of our knowledge we are the first to integrate these ideas with CBO. Furthermore, we also optimize 
the compilation steps in an instrumentation pipeline. 

\parttitle{Provenance Instrumentation}
Several systems such as \textit{DBNotes}~\cite{bhagwat2005annotation}, \textit{Trio}~\cite{aggarwal2009trio}, \textit{ORCHESTRA}~\cite{karvounarakis2013collaborative}, \textit{Perm}~\cite{glavic2013using}, \textit{LogicBox}~\cite{GA12}, \textit{ExSPAN}~\cite{ZS10}, and \textit{GProM}~\cite{arab2014generic} model provenance as annotations on data and capture provenance by propagating annotations. Most systems apply the \textit{provenance instrumentation} approach described in the introduction by compiling provenance capture and queries into a relational query language (typically SQL).
%
%
Thus, the techniques we introduce in this work are applicable to a wide range  of systems. 

\parttitle{Optimizing Provenance Computation and Storage}
Optimization of provenance has mostly focused on minimizing the storage size of provenance. Chapman et al.~\cite{CJ08a} introduce several techniques for compressing provenance information, e.g., by replacing repeated elements with references and discuss how to maintain such a storage representation under updates. Similar techniques have been applied to reduce the storage size of provenance for workflows that exchange data as nested collections~\cite{AB09}. A cost-based framework for choosing between reference-based provenance storage (the provenance of a tuple is distributed over several nodes) and propagating full provenance (full provenance is propagated alongside the tuple) was introduced in the context of declarative networking~\cite{ZS10}. 
%
%
This idea of storing just enough information to be able to reconstruct provenance through instrumented replay, has also been adopted for computing the provenance for transactions~\cite{arab2014generic,AG16a,AG17} 
and in the Subzero system~\cite{wu2013subzero}. 
Subzero switches between different provenance storage representations in an adaptive manner to optimize the cost of provenance queries.  
Amsterdamer et al.~\cite{AD12} demonstrate how to rewrite a UCQ query with inequalities into an equivalent query with 
provenance of minimal size. 
%
%
Our work is orthogonal to these approaches in that we try to minimize the time for on-demand provenance generation and queries over provenance instead of compressing provenance to minimize storage size. 
It would be interesting to integrate compact representations of provenance with CBO, e.g., choose among alternative compression methods. 
\section{Experiments}\label{sec:experiments}

Our evaluation focuses on measuring 1) the effectiveness of CBO in choosing the most efficient ICs and PATs, 2) the effectiveness of  heuristic application of PATs, 3) the overhead of heuristic and cost-based optimization, and 4) the impact of CBO search space traversal strategies on optimization and execution time.
All experiments were executed on a  machine with 2 AMD Opteron 4238 CPUs, 128GB RAM, and a hardware RAID with  4 $\times$ 1TB 72.K HDs in RAID 5 running commercial DBMS X (name omitted due to licensing restrictions).

To evaluate the effectiveness of our CBO 
vs.  
heuristic optimization choices, 
we compare the performance of instrumented queries generated by the CBO (denoted as \textbf{\textit{Cost}}) against queries generated by selecting a predetermined option for each choice point. Based on a preliminary study we have selected three choice points:
 1) using the \textbf{\textit{Window}} or  \textbf{\textit{Join}} method; 2) using \textbf{\textit{Filter\-Updated}} or \textbf{\textit{HistJoin}} and 3) choosing whether to apply PAT rule~\eqref{eq:duplicate-remove-set} (remove duplicate elimination). If CBO is deactivated, then we always remove such operators if possible. 
The application of the remaining PATs introduced in Sec.~\ref{sec:heuristic} turned out to be always beneficial in our experiments. 
Thus, these PATs are always applied as long as their precondition is fulfilled.
We consider two variants for each of method: activating heuristic application of the remaining PATs (suffix \textbf{\textit{Heu}}) or deactivating them (\textbf{\textit{NoHeu}}).  Unless noted otherwise, results were averaged over 100 runs.



\subsection{Datasets \& Workloads}


\parttitle{Datasets}
\underline{TPC-H}:
We have generated TPC-H benchmark datasets of size 10MB, 100MB, 1GB, and 10GB (SF0.01 to SF10). 
%
\underline{Synthetic}:
For the transaction provenance experiments we use a 1M tuple relation with uniformly distributed numeric attributes. 
We vary the size of the transactional history (this affects
performance, because the database has to store this history to enable time
travel which is used when capturing transaction provenance). Parameter $HX$
indicates $X\%$ of history, e.g., $H10$ represents $10\%$ history (100K tuples).
%
\underline{DBLP}:
This dataset consistes of 8 million co-author pairs 
extracted from DBLP (\url{http://dblp.uni-trier.de/xml/}). 



\parttitle{Simple aggregation queries}
This workloads computes the provenance of queries consisting solely of aggregation operations using the instrumentation technique based on the rewrite rules pioneered in Perm~\cite{glavic2013using} and extended in GProM~\cite{arab2014generic}. An aggregation query consisting of $i$ aggregation operations where each
aggregation operates on the result of the previous aggregation. The leaf operation accesses the TPC-H \texttt{part} table. Every aggregation groups the input on a range of primary key attribute values such that the last step returns the same number of results independent of $i$. 

\parttitle{TPC-H queries}
We have selected 11 
queries from the 22 TPC-H queries to evaluate optimization of provenance computations for complex queries. We use GProM's instrumentation approach to compute provenance.

\parttitle{Transactions}
We use the \textit{reenactment} approach of GProM~\cite{AG16a,AG17} to compute provenance for transactions. 
The transactional workload is run upfront (not included in the measured execution time) and provenance is computed retroactively.
We vary the number of
updates per transaction, e.g., $U10$ is a transaction with 10 updates. 
The tuples to be updated are selected randomly using the primary key of
the relation.  All transactions were executed under isolation
level~\texttt{SERIALIZABLE}. 



\parttitle{Provenance export}
We use the approach from~\cite{NX15} to translate a relational encoding of provenance (see Sec.~\ref{sec:intro}) into PROV-JSON. 
We export the provenance for a query over the TPC-H schema that is a foreign key join across relations nation, customer, and orders. 

\parttitle{Datalog provenance queries}
We use the approach described in~\cite{LS16} using the pipeline shown in Fig.~\ref{fig:DL-rewrite-approach}. The input is a non-recursive Datalog 
query 
$Q$ and a user question asking why (or why-not) a set of tuples is in the result of query $Q$.  
We use the DBLP co-author dataset for this experiment and the following queries. \textbf{Q1}: Returns authors which have co-authors that have co-authors.
\textbf{Q2}: Returns authors that are co-authors, but not of themselves (while this is semantically meaningless, it is a  way to test negation).
\textbf{Q3}: Return pairs of authors that are indirect co-authors, but are not direct co-authors.
\textbf{Q4}: Return start points of paths of length 3 in the co-author graph. For each query we consider multiple why questions that specify the set of results for which provenance should be generated. 
We use Qi.j to denote the $j^{th}$ why question for query Qi.




\begin{figure*}[p]

\begin{minipage}{1\linewidth}
\hspace{-10mm}
{
  \resizebox{0.6\linewidth}{!}{
  \begin{minipage}{0.9\linewidth}
  \centering
  \begin{tabular}{|c|r|r|r|r|r|} \hline 
\rowcolor[gray]{.9}  Queries & NoHeu+Join & Heu+Join & NoHeu+Window & Heu+Window & Cost+Heu\\ \hline 
SAgg 1G & 4.79 & 20.21 & 4.38 & 2.69 & \textbf{0.81} \\ \hline 
SAgg 10G & 44.06 & 524.78 & 42.62& 27.47 & \textbf{7.65} \\ \hline 
    \hline
TPC-H 1G & $+$173,053.17 & \textbf{199.62} & 173,041.27 & 250.18 &  235.79 \\ \hline 
    TPC-H 10G & $+$175,371.02 & \textbf{2,033.71} & 175,530.53 & 2,247.39 & 2,196.01 \\ \hline
  \end{tabular}
\end{minipage}
}
}
\hspace{-20mm}
{
  \resizebox{0.6\linewidth}{!}{
  \begin{minipage}{0.9\linewidth}
  \centering
  \begin{tabular}{|c|r|r|r|r|r|} \hline 
\rowcolor[gray]{.9}  Queries & NoHeu+Join & Heu+Join & NoHeu+Window & Heu+Window & Cost+Heu\\ \hline 
SAgg 1G & 1 & 3.927 & 0.946 & 0.600 & \textbf{0.261} \\ \hline 
SAgg 10G & 1 & 9.148 & 0.984 & 0.655 & \textbf{0.265} \\ \hline 
\hline
    TPC-H 1G & 1  & \textbf{0.187}  & 0.955 & 0.220 &  0.203\\ \hline 
TPC-H 10G & 1  & 0.198 & 0.975  & 0.180 & \textbf{0.174} \\ \hline
  \end{tabular}
\end{minipage}
}
}
\caption{Total runtime (\textbf{Left}) and average runtime (\textbf{Right}) per query relative to NoHeu+Join for \textit{SimpleAgg} and \textit{TPC-H} workloads}
\label{tab:overview-sum-avg-sagg-tpch}
\end{minipage}
\begin{minipage}{0.66\linewidth}
  \begin{minipage}[b]{0.5\linewidth}
  \includegraphics[width=1\linewidth,trim=0 50pt 0 100pt, clip]{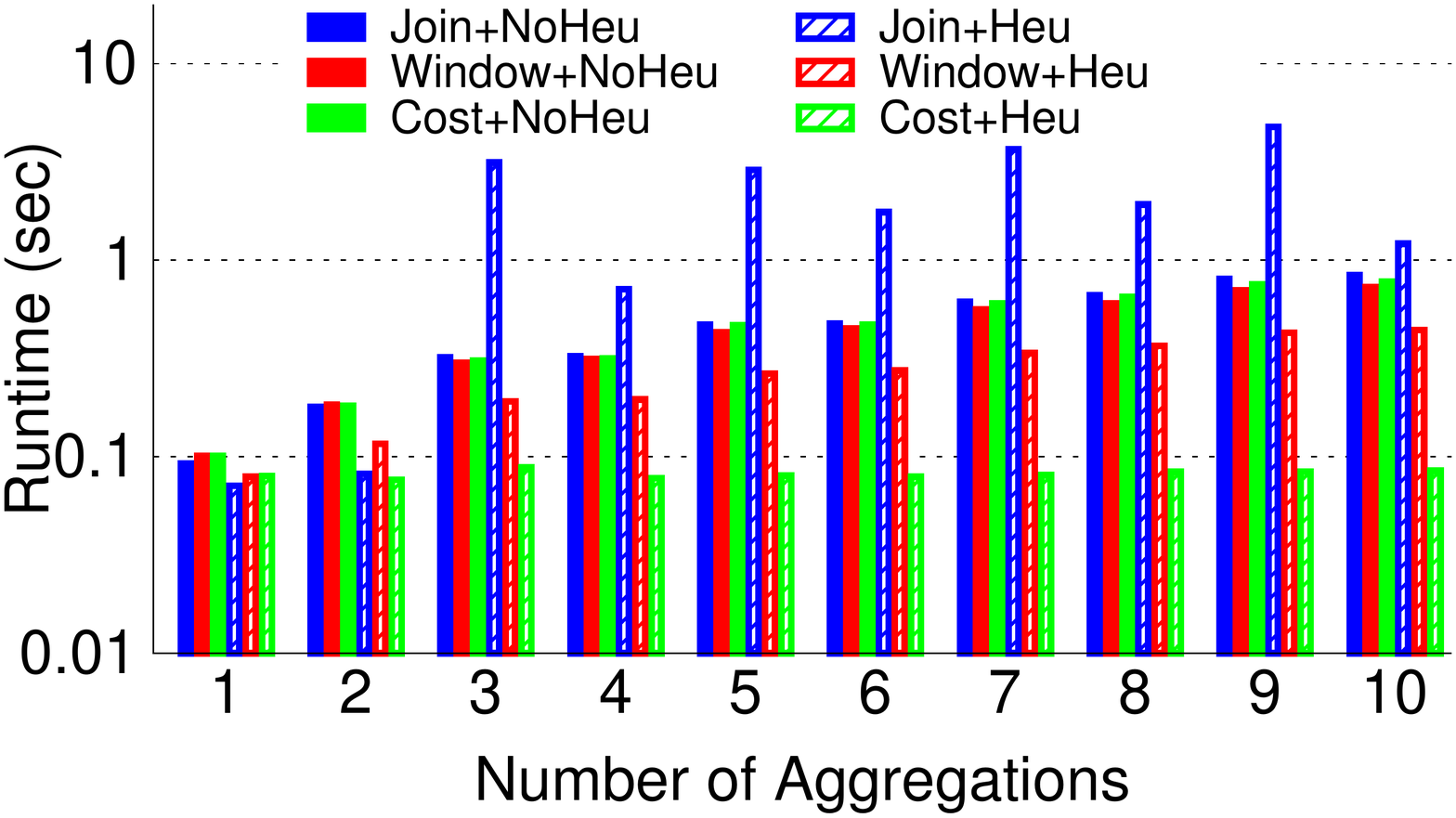}\\[-7mm]
  \caption{1GB \textit{SimpleAgg} Runtime}
  \label{fig:simpleAgg-comb-1GB}  
  \end{minipage}
  \begin{minipage}[b]{0.5\linewidth}
  \includegraphics[width=0.97\linewidth,trim=0 50pt 0 90pt, clip]{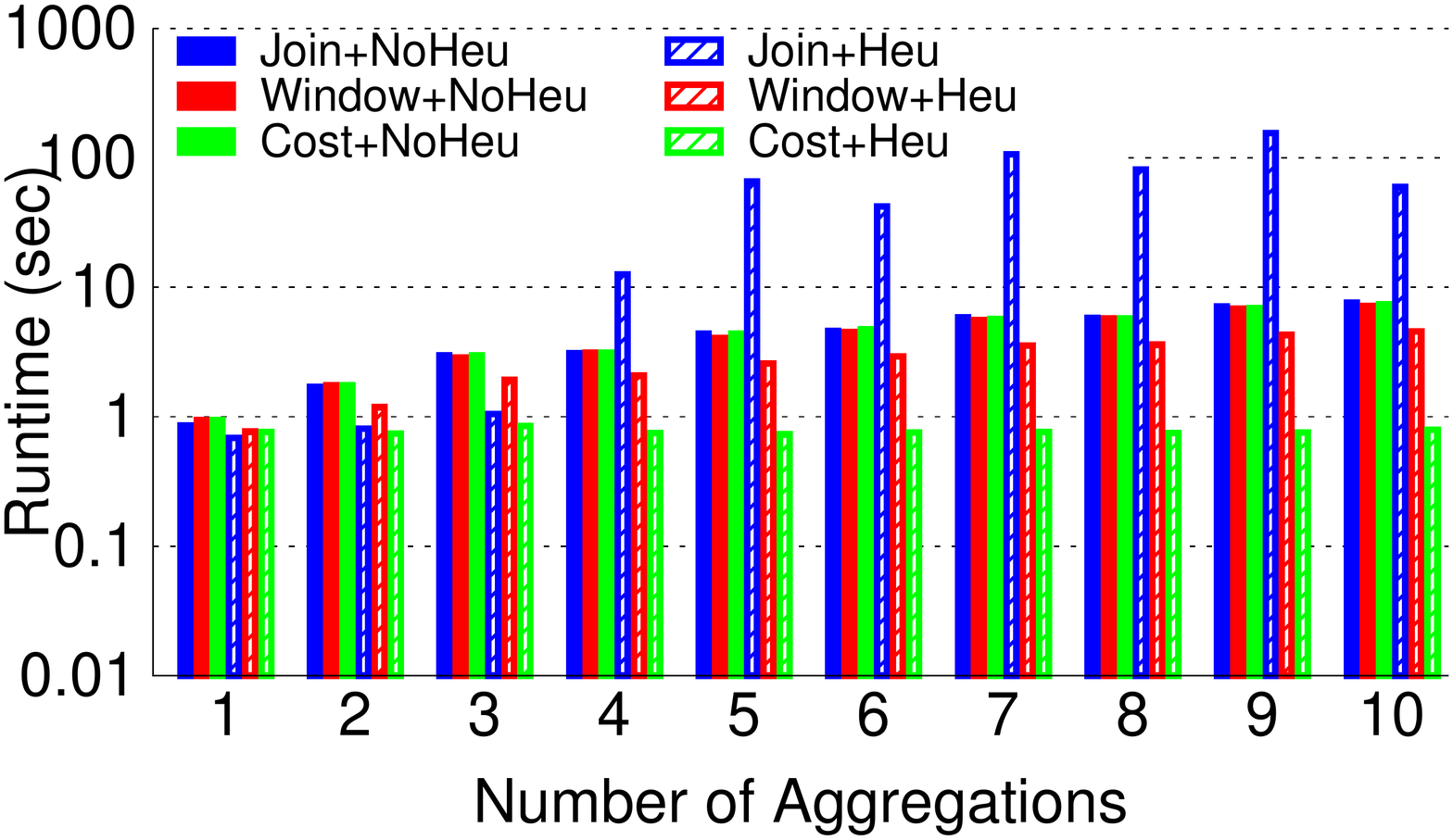}\\[-7mm] 
  \caption{10GB \textit{SimpleAgg} Runtime}
  \label{fig:simpleAgg-comb-10GB}
  \end{minipage}

  \begin{minipage}[b]{0.5\linewidth}
  \includegraphics[width=1\linewidth,trim=0 50pt 0 100pt, clip]{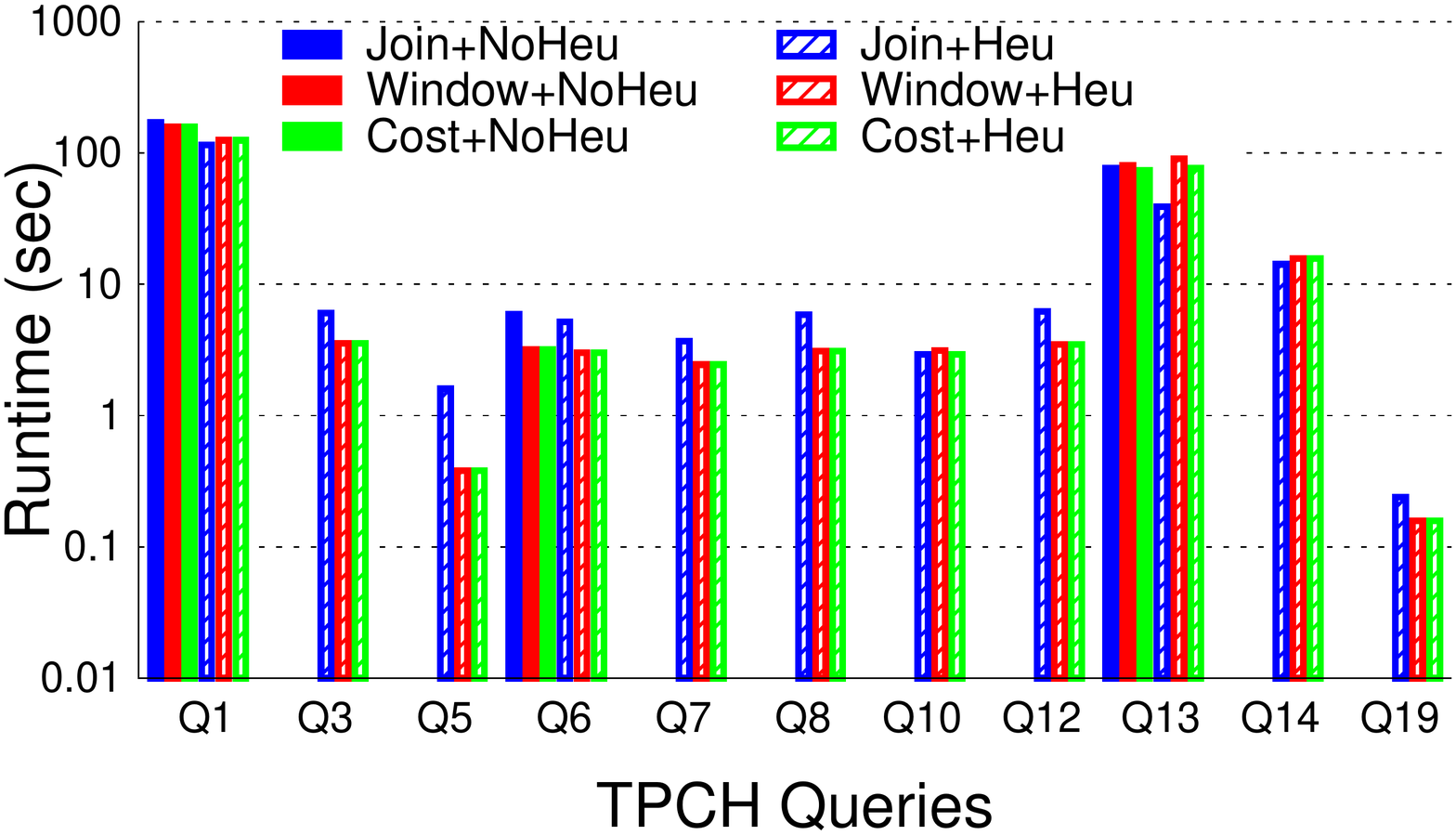}\\[-7mm]
  \caption{Runtime \textit{TPC-H} - 1GB}
  \label{fig:tpch-comb-1GB}  
  \end{minipage}
  \begin{minipage}[b]{0.5\linewidth}
  \includegraphics[width=1\linewidth,trim=0 50pt 0 100pt, clip]{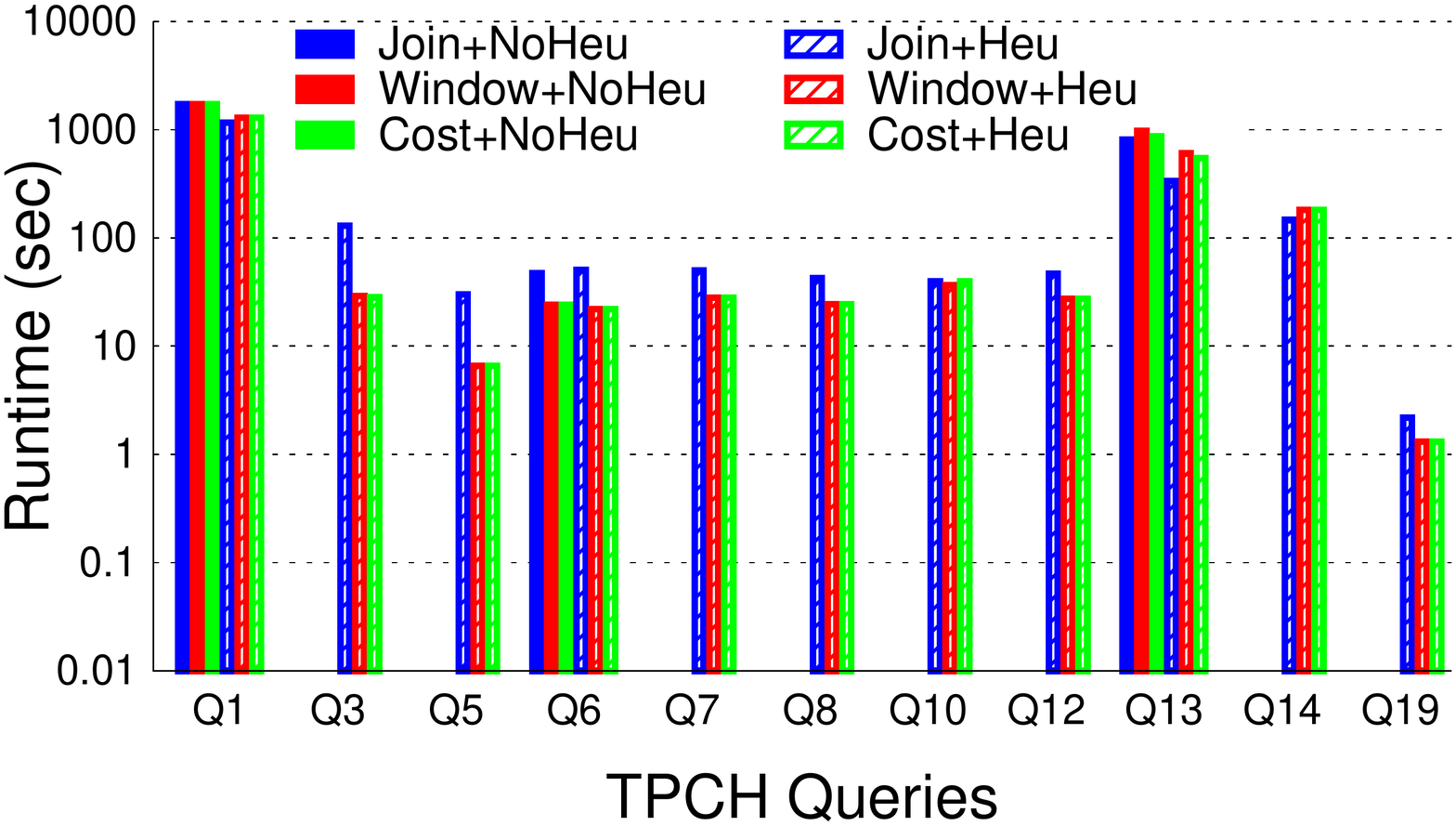}\\[-7mm] 
  \caption{Runtime \textit{TPC-H} - 10GB}
  \label{fig:tpch-comb-10GB}
\end{minipage}


  \begin{minipage}[b]{0.5\linewidth}
  \includegraphics[width=1\linewidth,trim=0 60pt 0 100pt, clip]{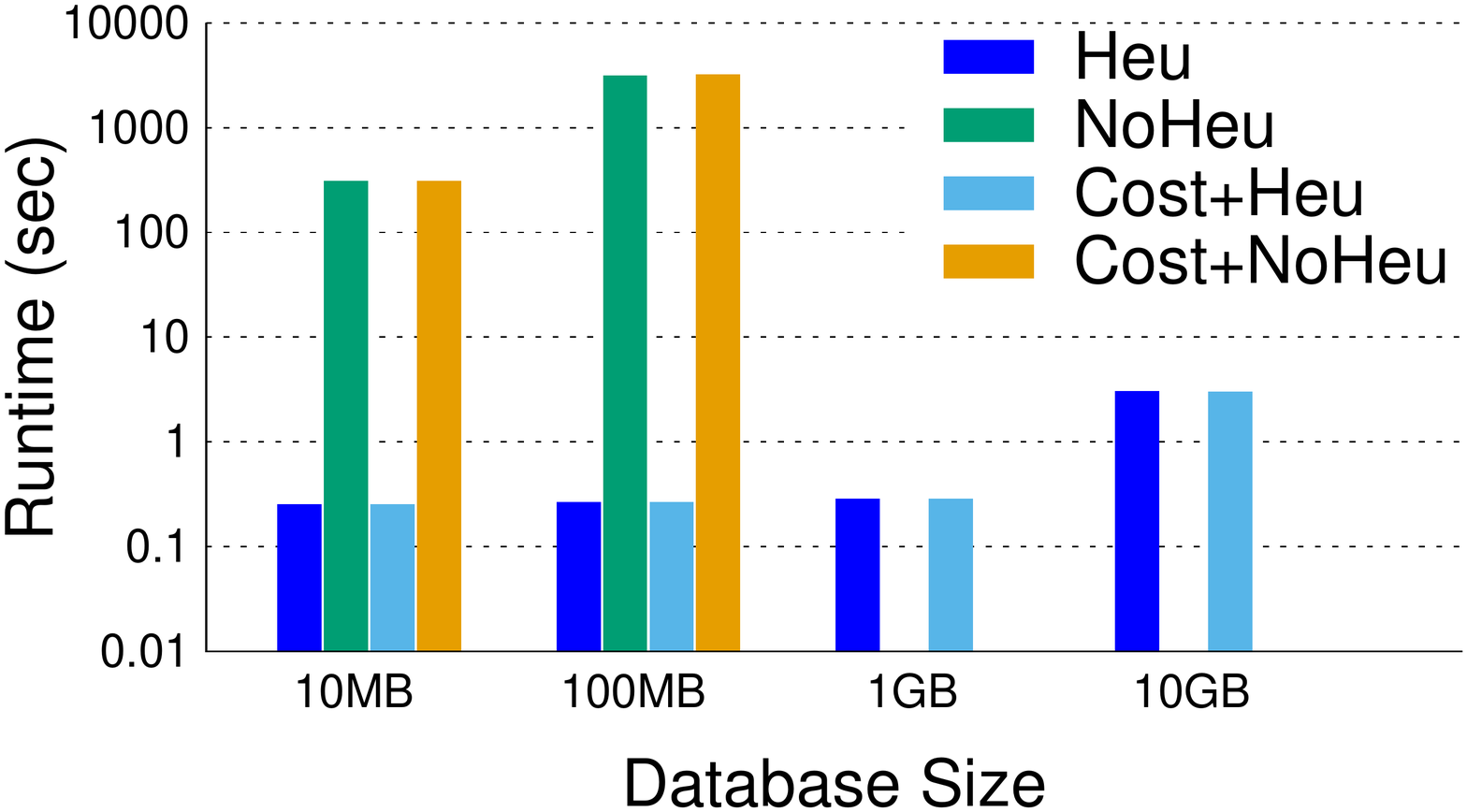}\\[-7mm]
  \caption{Provenance Export}
  \label{fig:export}  
  \end{minipage}
  \begin{minipage}[b]{0.5\linewidth}
  \includegraphics[width=1\linewidth,trim=0 60pt 0 100pt, clip]{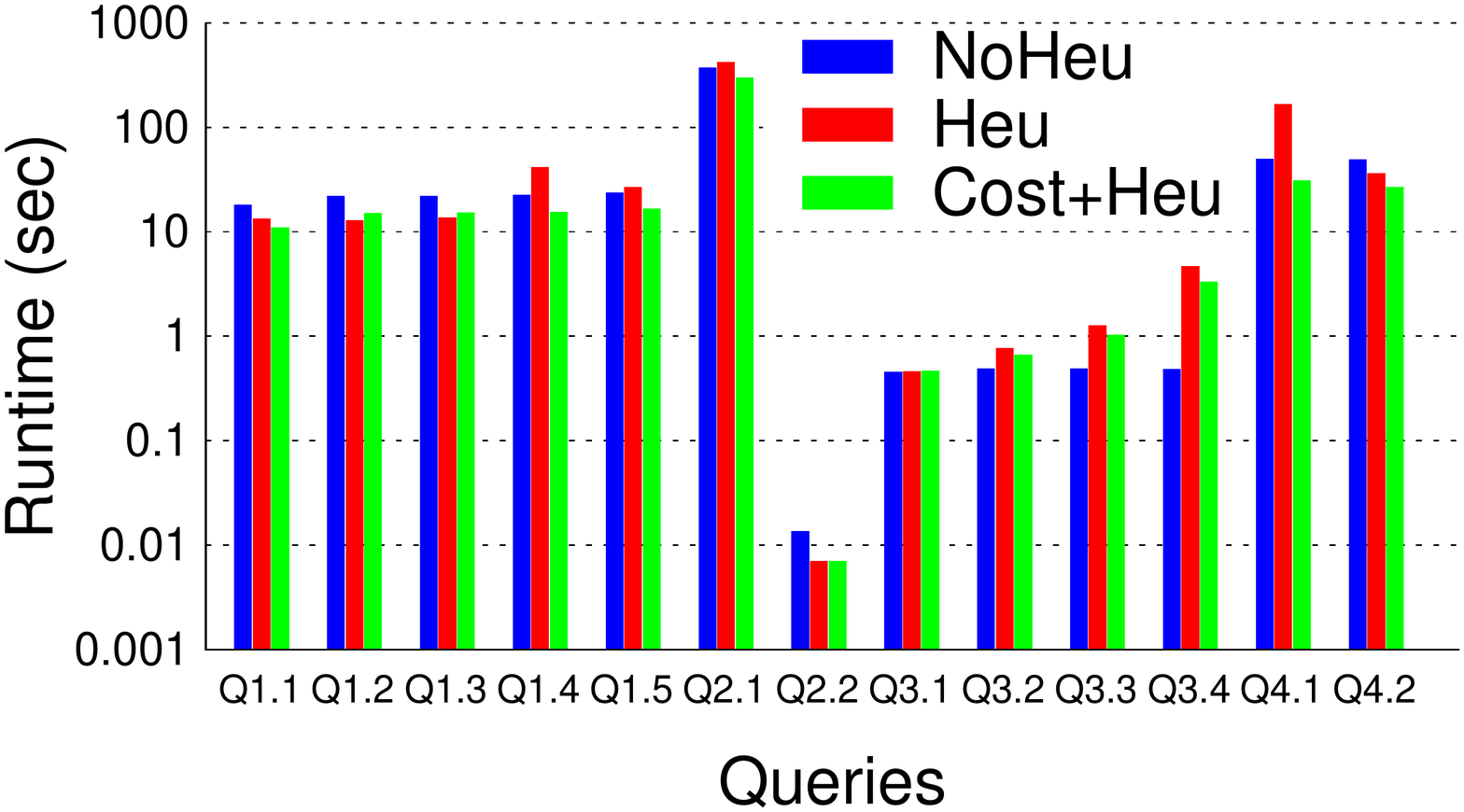}\\[-7mm] 
  \caption{Datalog Provenance}
  \label{fig:provenance-game}
  \end{minipage}

  
\begin{minipage}[b]{1\linewidth}  
\begin{minipage}[b]{0.48\linewidth}
\includegraphics[width=1\linewidth,trim=0 -10pt 0 30pt, clip]{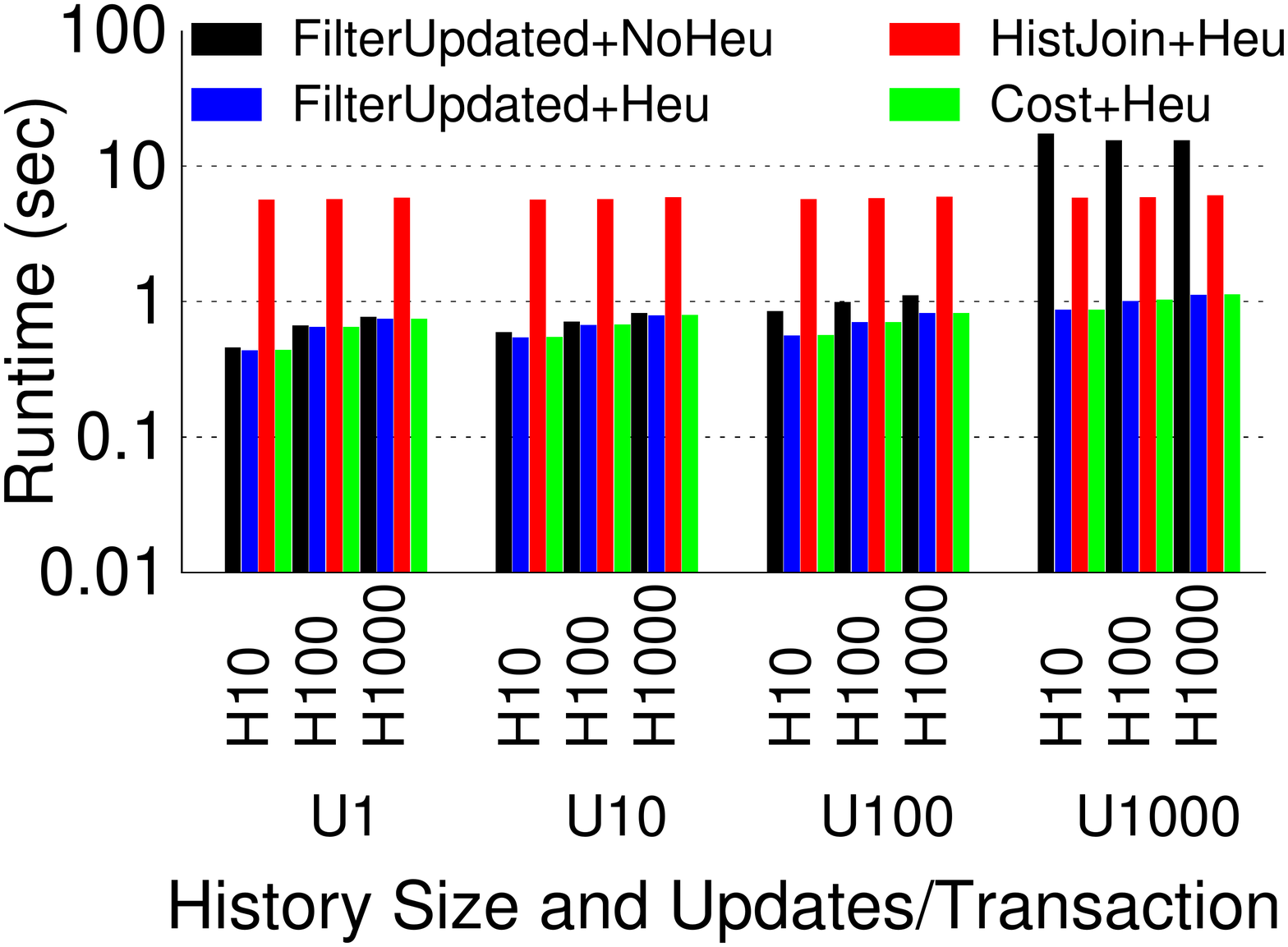}\\
\vspace{-7mm}
\end{minipage}
  \begin{minipage}[b]{0.48\linewidth}
\includegraphics[width=1\linewidth,trim=0 0pt 0 30pt, clip]{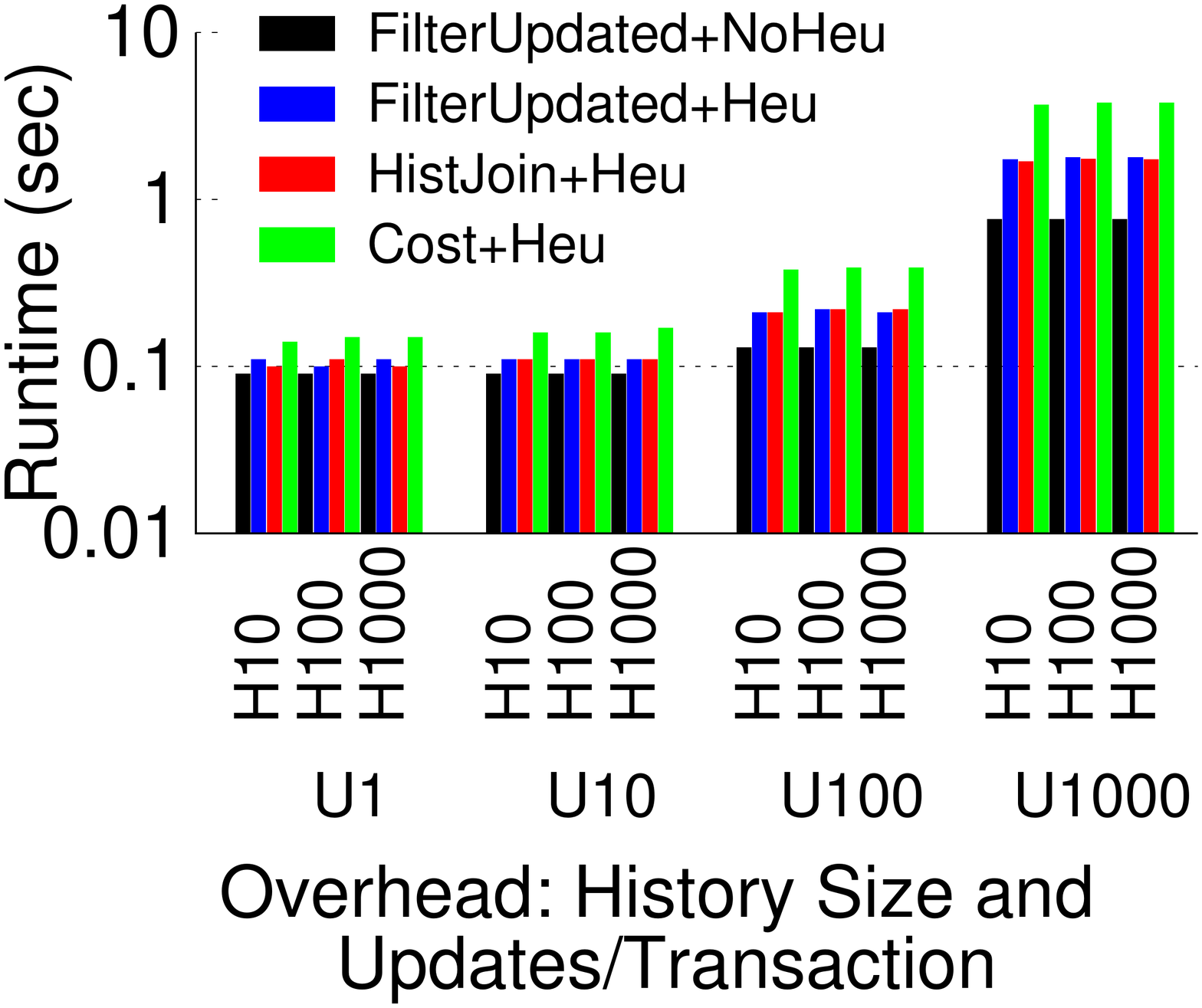}\\
  \vspace{-7mm}
\end{minipage}
  \caption{Transaction provenance - runtime and overhead}
  \label{fig:Transaction-provenance-runtime}
\end{minipage}


  \begin{minipage}{1\linewidth}
  \begin{minipage}[b]{0.49\linewidth}
  \includegraphics[width=1\linewidth,trim=0 50pt 0 100pt, clip]{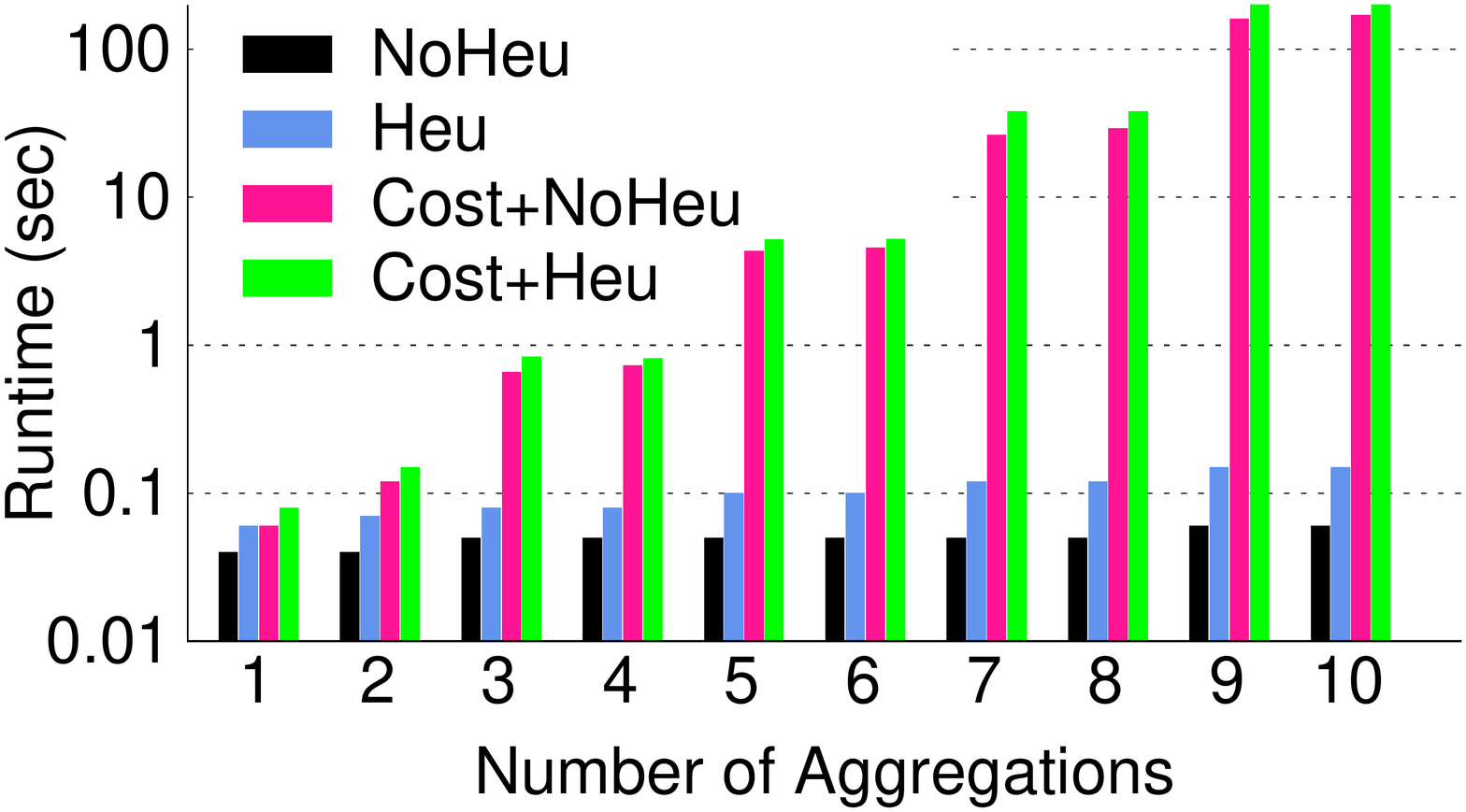}\\[-7mm]
  \label{fig:simple-agg-overhead}  
  \end{minipage}
  \begin{minipage}[b]{0.49\linewidth}
  \includegraphics[width=1\linewidth,trim=0 50pt 0 100pt, clip]{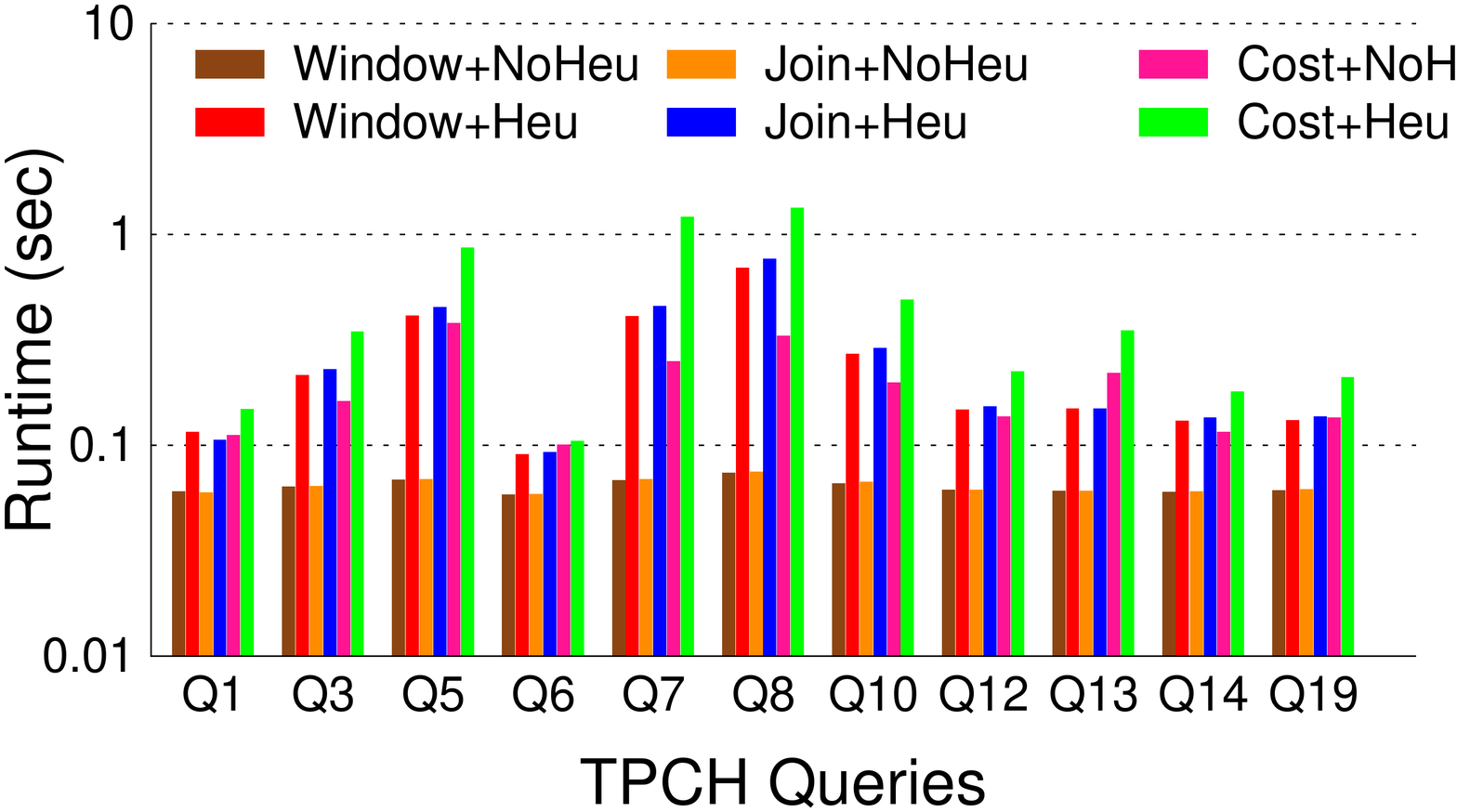}\\[-7mm] 
  \label{fig:tpch-overhead}
  \end{minipage}
  \caption{\textit{SimpleAgg} (\textbf{Left}) and \textit{TPC-H} (\textbf{Right}) Overhead}
  \label{fig:simagg-tpch-overhead}
\end{minipage}

\end{minipage}
\begin{minipage}{0.33\linewidth}
  \begin{minipage}{1\linewidth}
  \centering
\resizebox{1\linewidth}{!}{  
  \begin{minipage}{1.55\linewidth}
    \centering
  \begin{tabular}{|c|r|r|r|r|r|} \hline 
    \rowcolor[gray]{.9}   & \multicolumn{1}{|c|}{NoHeu}  & \multicolumn{1}{|c|}{NoHeu} & \multicolumn{1}{|c|}{Heu}  & \multicolumn{1}{|c|}{Heu}  & \multicolumn{1}{|c|}{Cost+Heu}\\
    \rowcolor[gray]{.9}   & \multicolumn{1}{|c|}{(Worst)} & \multicolumn{1}{|c|}{(Best)} & \multicolumn{1}{|c|}{(Worst)} & \multicolumn{1}{|c|}{(Best)} & \\
    \hline 
Min & 1.33 & 1.33 & \textbf{1.00} & \textbf{1.00} &  \textbf{1.00}\\ \hline 
Avg & 1,878.76 & 1,877.95 & 14.16 & 2.82 & \textbf{1.04} \\ \hline 
Max & $+$12,173.35& $+$12,173.35 & 68.63 & 7.80 & \textbf{1.18} \\ \hline 
  \end{tabular}
\end{minipage}
}\\
\caption{Min, max, and avg runtime relative  to the best method per workload aggregated over all workloads.}
\label{tab:overview-all}

\end{minipage}

\begin{minipage}{1\linewidth}
\resizebox{0.7\linewidth}{!}{  
  \begin{minipage}{0.98\linewidth}
    \centering
  \begin{tabular}{|c|r|r|r|r|} \hline 
  \rowcolor[gray]{.9}  Queries & FilterUpdated & HistJoin& FilterUpdated & Cost\\
  \rowcolor[gray]{.9}  Queries & +NoHeu & +Heu & +Heu & +Heu\\ \hline 
HSU/T & 55.11 & 69.50 &  \textbf{8.91} & \textbf{8.96} \\ \hline 
TAPU & 30.13 & 26.08 & \textbf{12.94} & \textbf{12.89} \\ \hline 
  \end{tabular}
\end{minipage}
}\\
\caption{Total workload runtime for transaction provenance}
\label{tab:overview-sum-transaction}
\end{minipage}

\begin{minipage}{1\linewidth}
 
\resizebox{0.7\linewidth}{!}{
  \begin{minipage}{0.88\linewidth}
  \begin{tabular}{|c|r|r|r|} \hline 
\rowcolor[gray]{.9}  Queries & NoHeu & Heu & Cost+Heu \\ \hline 
Export 10M & 310.49 &  \textbf{0.25}&  \textbf{0.25} \\ \hline 
Export 100M & 3,136.94 & \textbf{0.27} &  \textbf{0.26}  \\ \hline 
Export 1G & +21,600 & \textbf{0.28} &  \textbf{0.28}   \\ \hline 
Export 10G & +21,600 & \textbf{3.03} & \textbf{3.01}  \\ \hline \hline
Datalog Provenance & 583.96 & 736.50 & \textbf{437.75} \\ \hline 
  \end{tabular}
\end{minipage}
}\\
\caption{Total runtime for export and Datalog workloads}
\label{tab:overview-sum-export-gp}
\end{minipage}

  \begin{minipage}[b]{1\linewidth}
  \includegraphics[width=1\linewidth,trim=0 0pt 0 20pt, clip]{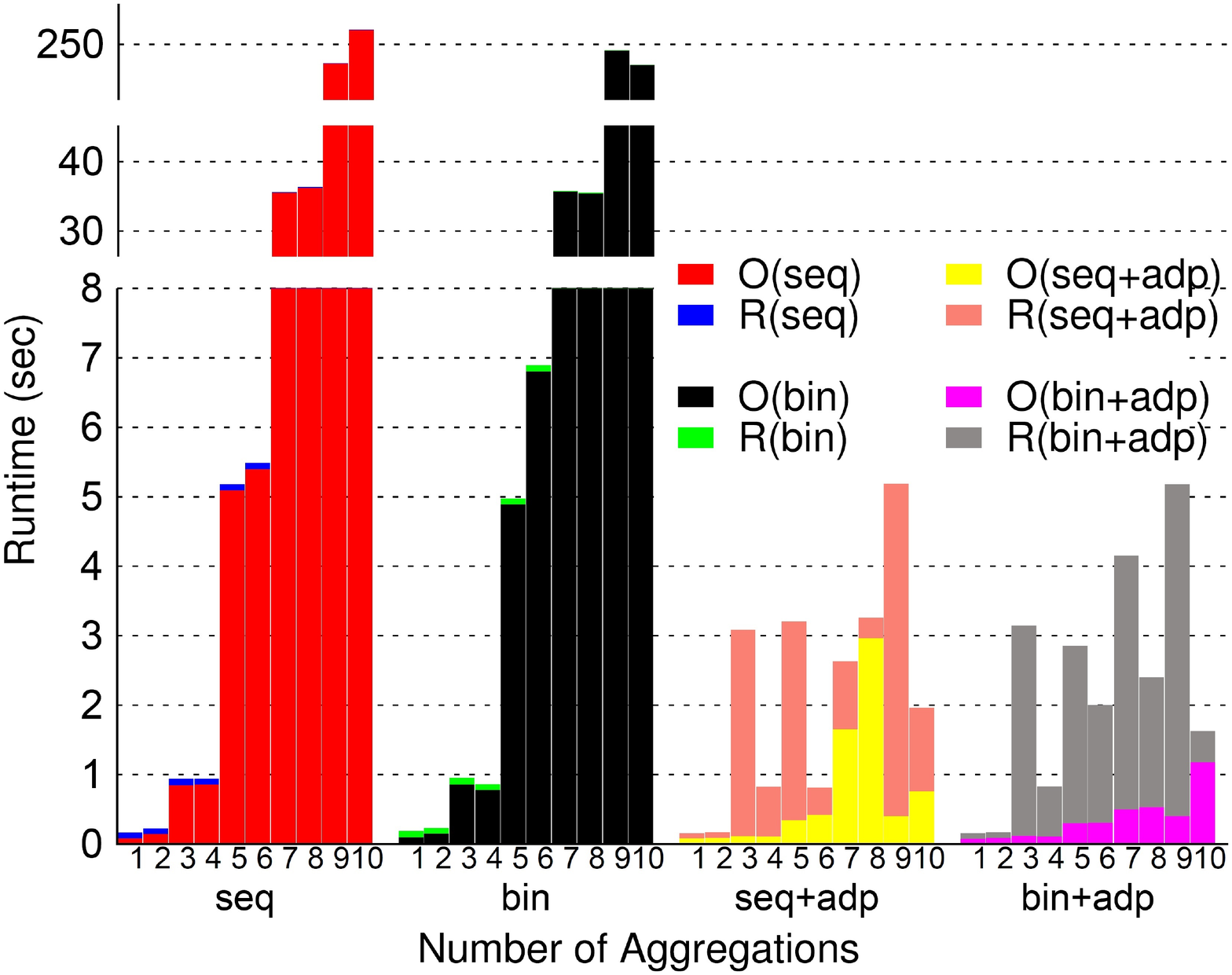}\\[-5mm]
  \caption{Optimization + runtime for Simple Agg. - 1GB}
  \label{fig:simAggs-all-stack}  
  \end{minipage}
  \begin{minipage}[b]{1\linewidth}
  \includegraphics[width=1\linewidth,trim=0 30pt 0 40pt, clip]{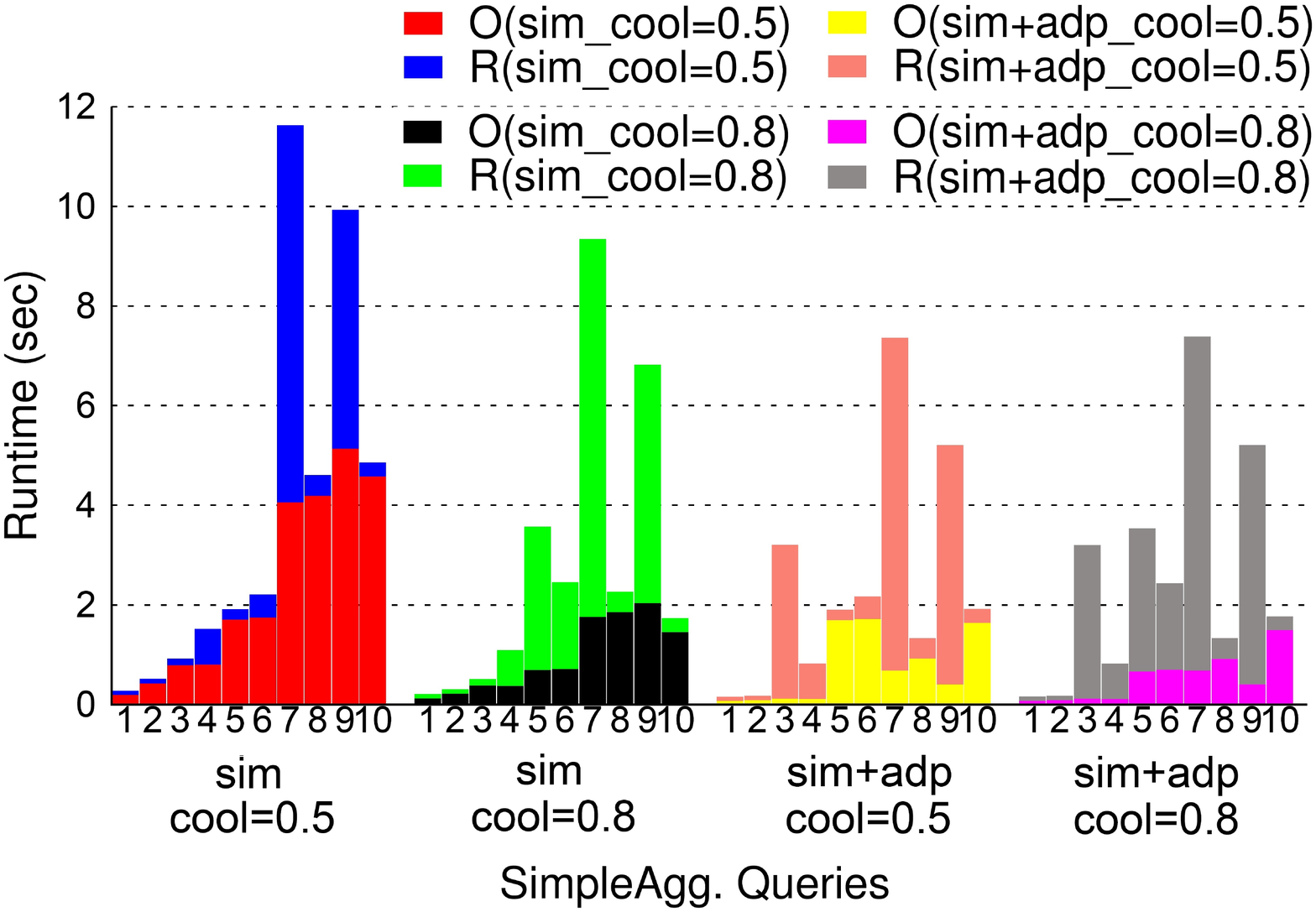}\\[-5mm] 
  \caption{Optimization + runtime for Simple Aggregation workload using Simulated Annealing - 1GB dataset}
  \label{fig:simple-agg-sim-annl}
  \end{minipage}

\end{minipage}
\end{figure*}

\subsection{Measuring Query Runtime}

\parttitle{Overview} 
Figure~\ref{tab:overview-all} shows an overview of our results. We show the average  runtime of each method relative to the best method per workload, e.g., if \textit{Cost} performs best for a workload then its runtime is normalized to 1. We  use relative overhead instead of total runtime over all workloads, because some of the workloads are significantly more expensive than others and, thus, comparing the results would be biased towards these workloads. For the \textit{NoHeu} and \textit{Heu} methods we report the performance of the best and the worst option for each choice point. For instance, for the \textit{SimpleAgg} workload the performance is impacted by the choice of whether the \textit{Join} or \textit{Window} method is used to instrument aggregation operators with \textit{Window} performing better (\textit{Best}). Numbers prefixed by a $'+'$  indicate that for this method some queries of the workload did not finish within the maximum time we have allocated for each query. Hence, the runtime for these cases should be interpreted as a lower bound on the actual runtime.  Compared with other methods, \textit{Cost+Heu} is on average only 4\% worth then the best method for the workload and has 18\% overhead in the worst case. Note that we confirmed that in all cases where an inferior plan was chosen by our CBO that was because of inaccurate cost estimations by the backend database. If we heuristically choose the best option for each choice point, then this results in a  178\% overhead over CBO on average. However, achieving this performance requires that the best option for each choice point is known upfront. 
The impact of bad choices on average increases runtime by a factor of $\sim$ 14 compared to CBO. These results also confirm the critical importance of our PATs since deactivating these transformations 
increases runtime by a factor of $\sim$ 1,800 on average and more than 12,000 in the worst case. 

\parttitle{Simple Aggregation Queries} 
We measure the runtime of computing provenance for the \textit{SimpleAgg} workload over the 1GB and 10GB TPC-H datasets varying the number of aggregations per query. 
The total workload runtime is shown in Fig.~\ref{tab:overview-sum-avg-sagg-tpch} (the best method is shown in bold). We also show the average runtime per query relative to the runtime of \textit{NoHeu+Join}. 
 Cost-based optimization significantly outperforms the other methods. The \textit{Window} method is more effective than the \textit{Join} method if a query contains multiple levels of aggregation. Our heuristic optimization improves the runtime of this method by about 50\%. The unexpected high runtimes of \textit{Join+Heu} are explained below. 
Fig.~\ref{fig:simpleAgg-comb-1GB} and \ref{fig:simpleAgg-comb-10GB} show the results for individual queries. Note that the y-axis is log-scale.
Activating heuristic optimizations improves performance in most cases, but for this workload the dominating factor is choosing the right method for instrumenting aggregations. 
The exception is the \textit{Join} method, where runtime increases when heuristic optimization is activated.  
We inspected the plans used by the backend DBMS for this case. A suboptimal join order was chosen for \textit{Join+Heu} based on inaccurate estimations of intermediate result sizes. For \textit{Join} the DBMS did not remove intermediate operators that blocked join reordering and, thus, executed the joins in the order provided in the input query which turned out to be more efficient in this particular case.
%
%
Consistently, the cost-based optimizer was either able to select \textit{Window} as the superior method (we confirmed this by inspecting the generated execution plan) or to outperform both \textit{Window} and \textit{Join} by instrumenting some of the aggregations in a query using the \textit{Window} and others with the \textit{Join} method.
\parttitle{TPC-H Queries} 
We compute the provenance of TPC-H queries to determine whether the results for simple aggregation queries
translate to more complex queries. 
The total workload execution time is shown in Fig.~\ref{tab:overview-sum-avg-sagg-tpch}. 
 We also show the average runtime per query relative to the runtime of \textit{NoHeu+Join}.
Fig.~\ref{fig:tpch-comb-1GB} and \ref{fig:tpch-comb-10GB} show the running time for each query for the 1GB and 10GB  datasets. Our CBO significantly outperforms the other methods with the only exception of \textit{Heu+Join}. Note that the runtime of  \textit{Heu+Join} for Q13 and Q14 is lower than \textit{Cost+Heu} which causes this effect.  
Depending on the dataset size and query, there are cases where the \textit{Join} method is superior and others where the \textit{Window} method is superior. The runtime difference between these methods is less pronounced than for \textit{SimpleAgg} presenting a challenge for our CBO. 
Except for Q13 which contains 2 aggregation operators, all other queries only contain one aggregation operator. The CBO was able to determine the best method to use in almost all cases. 
We confirmed that for the queries where we made an inferior choice, this was based on inaccurate cost estimates.
We also show the results for \textit{NoHeu}. However, only three queries finished in the allocated time slot of 6 hours (Q1, Q6 and Q13). Thus, the TPC-H results demonstrate the need for PATs and the robustness of our CBO in being able to choose the right instrumentation for a given query. 
\parttitle{Transactions}
We next compute the provenance of transactions executed over the synthetic dataset using the techniques introduced in~\cite{AG16a}. We vary the number of updates per
transaction ($U1$ up to $U1000$) and the size of the database's history
($H10$, $H100$, and $H1000$). 
The total workload runtime is shown in Fig.~\ref{tab:overview-sum-transaction}.  
The left graph in Fig.~\ref{fig:Transaction-provenance-runtime} shows detailed results. 
%
We compare the runtime of \textit{FilterUpdated} and \textit{HistJoin} (\textit{Heu} and \textit{NoHeu}) with \textit{Cost+Heu}. 
Our CBO choses \textit{FilterUpdated} as the better option for this workload. 
\parttitle{Provenance Export}
Fig.~\ref{fig:export} shows results for the provenance export workload for dataset sizes from 10MB up to 10GB (total workload runtime is shown in Fig.~\ref{tab:overview-sum-export-gp}). 
\textit{Cost+\-Heu} and \textit{Heu} both outperform \textit{NoHeu} demonstrating the key role of PATs for this workload. 
Our provenance instrumentations introduce window operators for enumerating intermediate result tuples which prevent the database from pushing selections and reordering joins.  \textit{Heu} outperforms \textit{NoHeu}, because \textit{Heu} determines that some of these window operators are redundant and can be removed (PAT rule~\eqref{eq:window-function}).
The CBO does not further improve the runtime for this workload, because this export query does not contain any aggregation and duplicate elimination operators, i.e., none of the choice points were hit. 

\parttitle{Why Questions for Datalog}
The approach~\cite{LS16} we use for generating provenance for 
Datalog queries with negation may produce queries which contain a large amount of duplicate elimination operators and shared subqueries. The heuristic application of PATs would remove all but the top-most duplicate elimination operator (rules~\eqref{eq:duplicate-remove} and~\eqref{eq:duplicate-remove-set}  in Fig.~\ref{fig:algebraic-rules}). However, this is not always the best option, because a duplicate elimination, while adding overhead, can reduce the size of inputs for downstream operators. Thus, as mentioned before we consider the application of rule 2 as an optimization choice in our CBO.
The total workload runtime and results for individual queries are shown in Fig.~\ref{tab:overview-sum-export-gp} respective Fig.~\ref{fig:provenance-game}. 
Removing all redundant duplicate elimination operators (\textit{Heu}) is not always better than removing none (\textit{NoHeu}). Our CBO (\textit{Cost+Heu}) has the best performance in almost all cases by choosing a subset of duplicate elimination operators to remove. Incorrect choices are again based on inaccurate cost estimation. 

\subsection{Optimization Time and CBO Strategies}


\parttitle{Simple Aggregation}
We show the optimization time of several methods in Fig.~\ref{fig:simagg-tpch-overhead} (left). 
Heuristic optimization (\textit{Heu}) results in an overhead of $\sim$50ms  compared to the time of compiling the provenance request without any optimization (\textit{NoHeu}) and this overhead is only slightly affected by the number of aggregations in the query. When increasing the number of aggregations,
the running time of \textit{Cost} increases more significantly because we have 2 choices for each aggregation, i.e., the plan space size is $2^{i}$ for $i$ aggregations. 
We have measure where time is spend during cost-based optimization and have determined that the majority of time is spend in costing SQL queries using the backend DBMS. Note that even though we did use the exhaustive search space traversal method for our CBO,  
the sum of optimization time and runtime for \textit{Cost} is still less than this sum for the \textit{Join} method.


\parttitle{TPC-H Queries}
In Fig.~\ref{fig:simagg-tpch-overhead} (right), we show the optimization time for TPC-H queries. Activating PATs results in $\sim$50ms overhead in most cases with a maximum overhead of $\sim$0.5s.  This is more than offset by the gain in query performance (recall that with \textit{NoHeu} only 3 queries finish within 1 hour for the 1GB dataset). CBO takes up to 3s in the worst case.
\parttitle{CBO Strategies}
We now compare query runtime and optimization time for the CBO search space traversal strategies introduced in Sec.~\ref{sec:cbo}.   
Recall that the \textit{sequential-leaf-traversal \textbf{(seq)}} and
\textit{binary-search-traversal \textbf{(bin)}} strategies are both exhaustive strategies. 
 \textit{Simulated Annealing} \textbf{(sim)} is the metaheuristic as introduced in Sec.~\ref{sec:traversal-strategies}.
We also combine these strategies with our \textit{adaptative \textbf{(adp)}} heuristic that limits time spend on optimization based on the expected runtime of the best plan found so far. 
 Fig.~\ref{fig:simAggs-all-stack} shows the total time (runtime (\textbf{R}) + optimization time (\textbf{O})) for the simple aggregation workload. We use this workload because it contains some queries with a large plan search space.  
 Not surprisingly,  the runtime of queries produced by \textit{seq} and \textit{bin} is better than \textit{seq+adp} and \textit{bin+adp} as \textit{seq} and \textit{bin} traverse the whole search space. However, their total time is much higher than \textit{seq+adp} and \textit{bin+adp} for larger numbers of aggregations. 
Fig.~\ref{fig:simple-agg-sim-annl} shows the total time of \textit{sim} with and without the \textit{adp} strategy for the same workload. We used cooling rates of 0.5 and 0.8 because they result in better performance than other rates that we have tested. The \textit{adp} strategy improves the total runtime in all cases except for the query with 3 aggregation operators.  




\section{Conclusions and Future Work}\label{sec:conclusion}

We present the first cost-based optimization framework for  provenance instrumentation and its implementation in GProM. 
The motivation for this work is that instrumented queries which capture provenance are often not successfully optimized, even by sophisticated database optimizers.
Our approach supports both heuristic and cost-based choices and is applicable to a wide range of instrumentation pipelines.
We have developed several provenance-specific algebraic transformation (PAT) rules which significantly improve performance as well as study instrumentation choices (ICs), i.e., alternative ways of realizing provenance capture.
%
%
Our experimental evaluation demonstrates that our optimizations improve performance by several order of magnitude for diverse provenance tasks.                                       
%
There are several interesting avenues of future work. We would like to improve the performance of CBO by making our optimizer aware of the structure of a query such that it can cache the best plan for a subquery.
Furthermore, we plan to use the CBO to select among alternative compressed and approximate provenance representations when capturing provenance.
\bibliographystyle{abbrv}
{\small\sloppy
\bibliography{2016-prov-optimizer}  
}
%
%

\end{document}